\documentclass[12pt]{article}

\usepackage{amssymb,mathtools,amsthm, amsmath, amsfonts}
\usepackage{authblk}
\usepackage{fullpage}
\usepackage{enumitem}
\usepackage{graphicx} % Required for inserting images
\usepackage{tikz-cd}
\usepackage{hyperref}
\usepackage[capitalize,nameinlink]{cleveref}
\usepackage{lipsum}
\usepackage{subcaption}
\usepackage{epstopdf}
\usepackage{algorithmic}
\ifpdf
  \DeclareGraphicsExtensions{.eps,.pdf,.png,.jpg}
\else
  \DeclareGraphicsExtensions{.eps}
\fi

\usepackage{soul}

\usepackage{amsfonts, amsmath, amssymb}
\usepackage{mathtools}
\usepackage{mathrsfs}
\usepackage{mathpazo}
\usepackage{upgreek}
\usepackage{bbm}
\usepackage{tikz-cd}

\newcommand{\R}{\mathbb{R}}

\usepackage{booktabs}
\usepackage{multirow}
\usepackage{longtable}
\usepackage{tabularx, makecell}

\usepackage{enumitem}

\usepackage{color, colortbl}
\usepackage{graphicx}
\usepackage{subcaption}
\usepackage{wrapfig}
\usepackage[font = small,labelfont=bf,textfont=it]{caption}
\usepackage{tcolorbox}												%Tcolorbox
\tcbuselibrary{listings,theorems}
\newtcbox{\linedbox}[1][red]{on line,
arc=0pt,outer arc=0pt,colback=#1!10!white,colframe=#1!50!black, boxsep=0pt,left=1pt,right=1pt,top=2pt,bottom=2pt, boxrule=0pt,bottomrule=1pt,toprule=1pt}
\newtcbox{\roundbox}[1][red]{on line, arc=7pt,colback=#1!10!white,colframe=#1!50!black, before upper={\rule[-3pt]{0pt}{10pt}},boxrule=1pt, boxsep=0pt,left=6pt,right=6pt,top=2pt,bottom=2pt}
\newtcbox{\mybox}[1]{title=#1}
\usepackage{tikz}													%Tikz
\usetikzlibrary{shapes,arrows,positioning,decorations.pathreplacing}

\usepackage{lastpage} 
\usepackage{fancyhdr} 								%Headers
\usepackage[protrusion=true,expansion=true]{microtype}
\usepackage{xargs}
\providecommand*{\Dom}[1]{{\rm dom}}							%Domain
\providecommand{\diag}{\operatorname{diag}}						%Diagonal

\newcommand{\tr}{\operatorname{Tr}}

\newtheorem{theorem}{Theorem}[section]
\newtheorem{remark}[theorem]{Remark}
\newtheorem{lemma}[theorem]{Lemma}
\newtheorem{corollary}[theorem]{Corollary}
\newtheorem{proposition}[theorem]{Proposition}
\newtheorem{definition}[theorem]{Definition}
\newtheorem{question}{Question}[section]
\newtheorem{example}[theorem]{Example}

\title{Random Walks, Conductance, and Resistance for the Connection Graph Laplacian}
\author[1,2]{Alexander Cloninger}
\author[2]{Gal Mishne}
\author[3]{Andreas Oslandsbotn}
\author[1]{Sawyer Jack Robertson}
\author[2]{Zhengchao Wan}
\author[2]{Yusu Wang}

\affil[1]{Department of Mathematics, University of California San Diego}
\affil[2]{Hal{\i}c{\i}o\u{g}lu Data Science Institute, University of California San Diego}
\affil[3]{Department of Informatics, University of Oslo}

\date{}

\begin{document}

\maketitle

\begin{abstract}
    We investigate the concept of effective resistance in connection graphs, expanding its traditional application from undirected graphs. We propose a robust definition of effective resistance in connection graphs by focusing on the duality of Dirichlet-type and Poisson-type problems on connection graphs. Additionally, we delve into random walks, taking into account both node transitions and vector rotations. This approach introduces novel concepts of effective conductance and resistance matrices for connection graphs, capturing mean rotation matrices corresponding to random walk transitions. Thereby, it provides new theoretical insights for network analysis and optimization.
\end{abstract}

\section{Introduction}
Effective resistance is a widely used distance metric in graph theory, applicable in numerous fields such as dimensionality reduction
\cite{ham2004kernel} and graph sparsification \cite{spielman2011graph, chu2020graph} to graph clustering \cite{alev2018graph, zhang2019detecting}. Effective resistance is closely related to graph random walks through the concept of commute times \cite{chandra1989electrical} and is known for its ability to capture cluster structures \cite{bozzo2013resistancedistance} and its robustness to noise compared to the graph geodesic \cite{gaura2016resistance}. Furthermore, Mémoli et al.~\cite{memoli2022persistent} showed how effective conductance, the reciprocal of effective resistance, can be used to define Cheeger constants. The effective resistance has been successfully applied in several domains, including bioinformatics \cite{forcey2020phylogenetic}, social network systems \cite{zhang2019detecting} and electronics \cite{doyle1984random, kocc2014impact,tauch2015measuring,  wang2015network, cavraro2018graph}. Its presence as a popular tool for analyzing graphs showcases its effectiveness and adaptability in tackling complex real-world problems.

Meanwhile, the increasing complexity of data has motivated the study of more complex graph structures, such as directed graphs, magnetic graphs, and connection graphs \cite{fanuel2018magnetic,singer2012vector}. Of particular note are magnetic graphs and their Laplacian matrices. There has been much recent interest \cite{fiorini2023sigmanet,zhang2021magnet} in using magnetic graph Laplacians for directed graphs by adding an angular phase to directed edges. This has found application in directed graph visualization \cite{fanuel2018magnetic, cloninger2017note}, community detection \cite{fanuel2017magnetic}, as well as in the development of novel graph neural network architectures \cite{fiorini2023sigmanet,zhang2021magnet}. Our focus in this work is on connection graphs, which are a generalization of signed or magnetic graphs, or alternatively, a specific case of voltage graphs from algebraic graph theory \cite{gross1974voltage}. Notably, connection graphs find widespread application in tackling angular synchronization problems \cite{singer2011angular, bandeira2013cheeger}, as well as in leveraging diffusion maps for high-dimensional dataset analysis \cite{singer2012vector}.

Consequently, a pivotal question arises:
\begin{question}
    Can the notions of effective resistance and effective conductance be extended to connection graphs, and if so, would such an extension yield practical benefits?
\end{question}

\paragraph{Related work}

An extension of the concept of effective resistance to connection graphs was first proposed by Chung et al. \cite{chung2014connection} within the context of connection graph sparsification. Their proposed technique utilizes the pseudoinverse of the connection Laplacian matrix in a direct generalization of the definition for classical graphs. This approach, however, has limitations: it applies only to edges, does not extend to all pairs of vertices, and shows discontinuity with respect to changes in graph signatures. 

In recent years, at least two other approaches have been proposed for generalizing effective resistance to other graph models. In \cite{sugiyama2023kron}, the authors propose an extension of effective resistance to directed graphs using the Kron reduction of the Laplacian matrix and concepts from the random walk on the graphs. In \cite{song2019extension}, the authors propose a definition of effective resistance between sets in a graph using the Schur complement of the Laplacian. Both approaches, at least implicitly, rely on graph boundary value problems. For a detailed study on traditional graphs, see \cite{chung2000discrete}. 

\paragraph{Our contributions}

Traditional definitions of electrical resistance and effective conductance are intricately tied to Poisson and Dirichlet problems \cite{jorgensen2008operator} as well as random walks on graphs \cite{doyle1984random,lovasz1993random}. Inspired by this, our first contribution is to extend Dirichlet problems to connection graphs and establish several fundamental properties, such as the maximum norm principle. Furthermore, considering random walks on connection graphs, we pinpoint the pivotal concept of mean path signatures. These signatures track the expected rotation experienced by a random walk along the edges. By linking this concept with the Dirichlet problem on connection graphs and devising algorithms to compute mean path signatures, we establish a solid groundwork for defining effective resistance and effective conductance on connection graphs, generalizing classical definitions.

In particular, these concepts and techniques allow us to derive the "effective conductance matrix"—a counterpart to the classical effective conductance in the context of connection graphs. Our definition preserves continuity w.r.t. changes to signatures and permits the evaluation of all vertex pairs. Moreover, we establish a direct link between the effective conductance matrix and mean path signatures, detailing how graph signatures affect effective conductance matrices. This work broadens the classical relationship between effective conductance and escape probability for random walks, extending it to connection graphs.

Finally, following the derivation of the effective conductance matrix, we propose a Poisson-type problem on connection graphs, creating a pathway to define the "effective resistance matrix." Despite the discontinuity of the effective resistance matrix, we introduce a scalar version that preserves continuity in response to changes in the underlying graph's signatures and enables evaluation of all vertex pairs. This scalar version is derived by analyzing the energy of the solution to the Poisson-type problem, further generalizing the classical scenario.
Significantly, our effective resistance matrix offers more insightful data than the scalar version alone. For instance, in the case of a cyclic graph, the scalar version fails to provide any signature information. Thus, in certain applications, it is imperative to consider the effective resistance matrix in its entirety rather than focusing solely on the scalar version.

\paragraph{Paper Organization} The remainder of this paper is organized as follows.
In \Cref{sec:basics-graphs-conn-laplacian}, we establish the notations for graphs and connection graphs, and review essential concepts and results, such as the switching equivalence of signatures.
Next, in \Cref{subsec:effective-resistance}, we provide a comprehensive overview of effective resistance from multiple perspectives.
\Cref{sec:random walks} considers the analysis of Dirichlet problems and random walks on connection graphs, which serve as the foundation for defining conductance and resistance for connection graphs.
Finally, in \Cref{sec:connection conductance and resistance}, we introduce the conductance matrix and in \Cref{subsec:connection-resistance-matrices} we introduce the resistance matrix for connection graphs. Throughout these sections, we present noteworthy findings, including their connections to random walks and the Dirichlet problem.
It is important to highlight that in \Cref{sec:scalar resistance matrix}, we utilize the aforementioned results to propose a scalar version of effective resistance for connection graphs, emphasizing its continuity with respect to changes in the underlying signature.

We provide demos and examples in our GitHub Repository\footnote{\url{https://github.com/sawyer-jack-1/connection-resistance-demo}}.

\section{Preliminary}
\label{sec:basics-graphs-conn-laplacian}

In this section we will set up notations and discuss some important results that we will use in the rest of the paper.

\subsection{Basics on Matrices and Graphs}
\label{subsec:graphs-matrices}

For positive integers $m$ and $n$, let $I_{n\times n}$ represent the $n\times n$ identity matrix, and $0_{m\times n}$ represent the $m\times n$ zero matrix. If $A\in\mathbb{R}^{n\times n}$ is a square matrix, $A^\dagger$ denotes its Moore-Penrose generalized inverse.

\paragraph{Schur complements} If $M = \left[\begin{array}{c|c}
    A & B \\
    \hline
    C & D
\end{array}\right]$ is a block matrix, where $A,D$ are square matrices, we define the (generalized) \emph{Schur complements} $M/A$ and $M/D$ by the formulas
\begin{align}
    M/A = D - C A^\dagger B,\hspace{.5cm}M/D = A-BD^\dagger C.
\end{align}
We note that the Schur complement can be defined for any principal submatrix of $M$.

One important property of the Schur complement, known as the \emph{quotient identity} \cite{crabtree1969identity}, will be utilized later in our analysis. Consider a block matrix $M$ as described above, and assume that the submatrix $D= \left[\begin{array}{c|c}
    E & F \\
    \hline
    G & H
\end{array}\right]
$ is also structured as a block matrix, with $E$ and $H$ representing square matrices. If both $D$ and $H$ are invertible, we have the following relationship:   
    \begin{equation}\label{eq:quotient}
    M/D=(M/H)/(D/H).
    \end{equation} 
In this equation, we implicitly utilize the fact that $D/H$ is a block submatrix of $M/H$.

\paragraph{Graphs}
We consider graphs $G = (V,E,W)$ where $V = \{1,2,\dotsc,n\}$ is a finite set of vertices, $E\subset{V\choose 2}$ is a collection of undirected edges, and $W = (w_{ij})_{i,j\in V}$ is an edge weight matrix with $w_{ij}>0\iff \{i,j\}\in E$.
We exclude multiple edges and self-edges. Additionally, we define the set of \emph{oriented edges} as
$E^\text{or} := \left\{(i,j),(j,i):i,j\in V, i\sim j\right\}.$ The degree of each $i\in V$ is denoted $\deg(i) := \sum_{j\sim i}w_{ij}$. \emph{Throughout the paper, we assume that graphs are \textbf{connected} unless otherwise stated.} 

The \textit{degree matrix} $D = \diag{(\deg(1),\dotsc,\deg(n))}\in\R^{n\times n}$ is the diagonal matrix whose elements are the degrees of $i\in V$.
The \textit{Laplacian matrix} $L$ of $G$ is defined by the equation $L = D-W$. See Chung \cite{chung1997spectral} for properties of the Laplacian matrix. 

\subsection{Connection Graphs} 
\label{subsec:connection-graph-basics}

Let $d\geq 1$ be a positive integer. We let $\mathsf{O}(d)$ denote the group of $d\times d$ orthogonal matrices, i.e., $\mathsf{O}(d) := \left\{O\in\mathbb{R}^{d\times d}:O^\mathrm{T}O = I_{d\times d}\right\}.$
A $d$-dimensional \textit{connection} (or a \textit{signature}) on a graph $G$ is a map $\sigma:E^{\text{or}}\rightarrow \mathsf{O}(d)$ which satisfies $\sigma_{ji} = \sigma_{ij}^{-1} = \sigma_{ji}^\mathrm{T}$ for each $(i,j)\in E^{\text{or}}$. The pair $(G,\sigma)$ is called a \textit{connection graph}. To distinguish between a graph $G$ and a connection graph $(G,\sigma)$ we might call $G$ a \textit{classical graph} and refer to $G$ as the \textit{underlying graph} of  $(G,\sigma)$.

The connection graph is known by different names depending on the value of $d$. When $d=1$, $(G,\sigma)$ is also called a \emph{signed graph}. It finds applications in social networks and voter models \cite{cartwright1956structural, li2015voter}. When $d=2$ and $\mathsf{SO}(2)$ is considered instead of $\mathsf{O}(2)$, $(G,\sigma)$ is called a \emph{magnetic graph}. These have applications from physics \cite{lieb1993fluxes}, the visualization of directed graphs \cite{fanuel2018magnetic}, and the angular synchronization problem \cite{singer2011angular}.

The \textit{connection Laplacian matrix} of a connection graph $(G,\sigma)$ is the $nd\times nd$ block matrix $\mathcal{L} $ (we sometimes write $\mathcal{L}^\sigma$ to emphasize the signature) defined as follows:
\begin{equation}\label{def:connection-laplacian}
    \mathcal{L} := \left[\begin{array}{c|c|c|c}
        \mathcal{L}_{11} & \mathcal{L}_{12} & \dotsc & \mathcal{L}_{1n} \\
        \hline
        \mathcal{L}_{21} & \mathcal{L}_{22} & \dotsc & \mathcal{L}_{2n} \\
        \hline
        \vdots           & \vdots           & \ddots & \vdots           \\
        \hline
        \mathcal{L}_{n1} & \mathcal{L}_{n2} & \dotsc & \mathcal{L}_{nn}
    \end{array}\right]
\text{ where }\mathcal{L}_{ij} :=\begin{cases}
        \deg(i) I_{d\times d}  & \text{ if }i=j     \\
        -w_{ij}\sigma_{ij} & \text{ if }i\sim j \\
        0_{d\times d}      & \text{ otherwise}
    \end{cases},\hspace{.2cm}  1\leq i,j \leq n.
\end{equation}

Clearly, the connection Laplacian matrix is a symmetric, positive semidefinite matrix. The matrix is particularly useful for analyzing vector valued functions $f:V\rightarrow\R^d$ defined on the underlying graph. We collect such functions into the linear space
$\ell^2 (V;\R^d) = \{f:V\rightarrow\R^d\}$
equipped with the Hilbert-Schmidt inner product
\begin{align}\langle f, g\rangle_{\ell^2 (V;\R^d)} = \tr(g^\mathrm{T} f),\hspace{.2cm}f,g\in \ell^2 (V;\R^d),\end{align}
where we identify each $f\in\ell^2 (V;\R^d)$ as a column vector in $\R^{nd}$ in the canonical way.
By applying the connection Laplacian matrix to any $f\in\ell^2 (V;\R^d)$, we obtain a new function $\mathcal{L}f:V\rightarrow\R^d$ which can be written explicitly as follows.
  \[(\mathcal{L}f)(i) = \sum_{j\sim i}w_{ij}\left(f(i)-\sigma_{ij}f(j)\right).\]
  One can also explicitly write the quadratic form as follows.
              \begin{align}\label{eq:connection-dirichlet-energy-vector-case}
                  f^\mathrm{T}\mathcal{L}f = \sum_{\{i,j\}\in E}w_{ij}(f(i)-\sigma_{ij}f(j))^\mathrm{T}(f(i)-\sigma_{ij}f(j)).
              \end{align}

\paragraph{Consistency} A connection graph $(G,\sigma)$ is said to be \textit{consistent} (and \emph{inconsistent} if otherwise) if for every directed cycle $i_0,\dotsc,i_{n}=i_0$ in $G$ it holds
\begin{align}\label{def:consistent-sig}
    \prod_{\ell=0}^{n-1} \sigma_{i_\ell i_{\ell+1}} = I_{d\times d}.
\end{align}

There are several equivalent criteria for a graph to be consistent, some of which we highlight in the following lemma; a proof can be found in \cite[Theorem 1]{chung2014connection}.

\begin{lemma}\label{lemma:balance-characterizations}
    Let $(G,\sigma)$ be a connection graph. $G$ is consistent if and only if
    \begin{enumerate}
        \item 0 occurs as an eigenvalue of $\mathcal{L}$ with multiplicity exactly $d$ times the number of connected components of $G$,
        \item The eigenvalues of $\mathcal{L}$ are exactly those of $L$ where each occurs with multiplicity $d$,
        \item There exists a map $\tau:V\rightarrow \mathsf{O}(d)$ such that
              \[\tau(i)^{-1}\sigma_{ij} \tau(j) = I_{d\times d}\text{ for each }(i,j)\in E^{\text{or}}.\]
    \end{enumerate}
\end{lemma}

\begin{remark}[Signatures of paths]\label{rmk:path-independence}
    If $(G,\sigma)$ is consistent, a straightforward consequence of \Cref{def:consistent-sig} is that for any (not necessarily adjacent) nodes $i,j\in V$ we may define $\sigma_{ij}$ by taking a path $(i_0=i,\dotsc,i_{m +1}= j)$ and writing
    $\sigma_{ij} = \prod_{\ell=1}^m\sigma_{i_\ell i_{\ell+1}}.$
    This definition is independent of the choice of path and is hence well defined.
\end{remark}

\subsection{Equivalence and Decomposition of Signatures}\label{subsec:equivalence-of-connection-graphs}

Different signatures on a graph $G$ may yield identical spectra for their respective connection Laplacians. This encourages a deeper exploration of signature structures. One particular important notion is that of (switching) equivalence of signatures from \cite{liu2019curvature} which we describe in a slightly different way below.

\begin{definition}[Switching equivalence between signatures]\label{defn:equivalence-signatures}
    Let $\sigma, \tau$ be two fixed signatures on a graph $G$. Then $\sigma, \tau$ are said to be \emph{(switching) equivalent}, denoted $\sigma\cong\tau$, if there exists a map $f:V\rightarrow \mathsf{O}(d)$ such that for any oriented edge $(i,j)\in E^\mathrm{or}$ it holds
    $f(i)\sigma_{ij} = \tau_{ij}f(j).$
    The map $f$ is called a \emph{switching map}.
\end{definition}

The map $f$ assigns an orthonormal basis in $\mathbb{R}^d$ to each node. At each edge, the orthogonal transformations defined by $\sigma$ and $\tau$ become equivalent when coordinates are changed into the bases given by $f$.

It is straightforward to show that $\cong$ defines an equivalence relation on the class of signatures defined on any one graph $G$. 
The next fact follows immediately from the transitivity of $\cong$, the preceding remarks, and \Cref{lemma:balance-characterizations}.

\begin{proposition}\label{prop:equivalence connection graph}
    Given any two consistent signatures $\sigma$ and $\tau$, they are equivalent. A consistent signature and an inconsistent signature are not equivalent.
\end{proposition}

The following result shows that one can significantly simplify the graph signature through a spanning tree of the underlying graph. The version of the following result with $\mathsf{U}(1)$ signatures has been mentioned in passing in \cite[Section 4.2]{liu2019curvature}.

\begin{lemma}[Spanning tree simplification]\label{lemma:spanning tree}
    Given a connected graph $G$, let $T$ be a spanning tree. Then, any signature $\sigma$ is equivalent to a signature $\sigma^T$ such that $\sigma^T_{i,j}=I_{d\times d}$ for any edge $(i,j)\in T$ and that $\sigma^T_{i,j}$ depends implicitly on $\sigma$ for any edge $(i,j)\notin T$.
\end{lemma}

\begin{proof}
 Define $f:E^{\text{or}}\rightarrow \mathsf{O}(d)$ as follows. Let $1\in V$ be a fixed distinguished vertex and set $f(1) = I_{d\times d}$. For each $i\in V$ let $P_i=(1=i_1, i_2,\dotsc, i_k = i)$ be the unique path contained in the spanning tree $T$ from $1$ to $i$. Define $f(i) = \prod_{\ell=1}^{k-1}\sigma_{i_\ell i_{\ell+1}}$ for every other $i\in V$. Then, setting $\sigma^T_{ij} = f(i)\sigma_{ij}f(j)^T$, the claim follows.
\end{proof}

The next result illustrates the relationship between connection Laplacian matrices of two equivalent signatures.

\begin{lemma}[{\cite[Equation (1.13)]{liu2019curvature}}]\label{prop:equivalence connection Laplacian}
    Let $f:V\rightarrow \mathsf{O}(d)$ be a switching map between signatures $\sigma$ and $\tau$.
    Then,
    $F\mathcal{L}^\sigma = \mathcal{L}^\tau F,$
    where $F\in\mathbb{R}^{nd\times nd}$ is the diagonal block matrix whose $i$-th $d\times d$ block is $f(i)$.
\end{lemma}

\paragraph{Direct sum of signatures}
Consider a $d$-dim signature $\sigma$ and a $d'$-dim signature $\sigma'$ on $G$.
We define the \textit{direct sum of $\sigma$ and $\sigma'$}, denoted by $\sigma\oplus\sigma'$, as follows:
\[\forall (i,j)\in E^\mathrm{or},\, (\sigma\oplus\sigma')_{ij}:=\sigma_{ij}\oplus\sigma'_{ij}  = \begin{bmatrix}
        \sigma_{ij}    & 0_{d\times d'} \\
        0_{d'\times d} & \sigma'_{ij}
    \end{bmatrix}.\]

\begin{remark}\label{rmk:block-decomposition-direct-sum}
    For  $\mathcal{L}^{\sigma\oplus\sigma'}$, we have that for any $i,j\in V$,
       $ \mathcal{L}^{\sigma\oplus\sigma'}_{ij} = \begin{bmatrix}
            \mathcal{L}^\sigma_{ij} & 0_{d\times d'}             \\
            0_{d'\times d}          & \mathcal{L}^{\sigma'}_{ij}\end{bmatrix}.$
    Therefore, by permutation of the rows and columns of $\mathcal{L}^{\sigma\oplus\sigma'}$, one has that $\mathcal{L}^{\sigma\oplus\sigma'}$ is similar to the matrix $\mathcal{L}^\sigma\oplus \mathcal{L}^{\sigma'}=\begin{bmatrix}
            \mathcal{L}^\sigma & 0                     \\
            0                  & \mathcal{L}^{\sigma'}\end{bmatrix}$.
\end{remark}

A $d$-dim signature $\sigma$ is called \textit{decomposable} if it is equivalent to the direct sum of two signatures with dimensions greater than 0. Otherwise, we call $\sigma$ \emph{indecomposable}\footnote{While previous research \cite{liu2019curvature} has examined the decomposition of connection Laplacians using group representation theory, we find the aforementioned explanation of signature decomposition to be crucial for comprehending our subsequent findings.}.
For instance, any $1$-dimensional signature is indecomposable, while consistent signatures with $d>1$ are always decomposable.

\begin{example}\label{ex:consistent signature}
    Let $\iota^1$ denote the $1$-dim identity signature.
    If a $d$-dim signature $\sigma$ is \emph{consistent}, then
        $\sigma \cong \bigoplus_{i=1}^d\iota^1.$
    This follows from \Cref{prop:equivalence connection graph} and \Cref{prop:equivalence connection Laplacian}.
\end{example}

The example above motivates us to focus on the case of inconsistent signatures.

\begin{theorem}\label{thm:decomposition}
    Let $\sigma$ be an \emph{inconsistent} $d$-dim signature. Let $\rho$ denote the nullity of $\mathcal{L}^\sigma$. Then, there exists a $(d-\rho)$-dim signature $\tau$ such that $\mathcal{L}^\tau$ is invertible and that
    \begin{equation}\label{eq:decomposition of signature}
        \sigma\simeq (\bigoplus_{i=1}^\rho\iota^1)\oplus \tau.
    \end{equation}
\end{theorem}

The proof technique below is similar to the proof of \cite[Theorem 1]{chung2014connection}. Besides, it is worth noting that the proof itself indicates an algorithm for the finding the decomposition in \Cref{eq:decomposition of signature}.
\begin{proof}
    
    Let $f_1,\ldots,f_\rho:V\to\R^d$ be independent eigenvectors of $\mathcal{L}^\sigma$ corresponding to the $0$ eigenvalues. Assume that $\langle f_l,f_l\rangle = |V|$ and $\langle f_l,f_k\rangle = 0$ for $l\neq k$.
    By \Cref{eq:connection-dirichlet-energy-vector-case}, we have that for any oriented edge $(i,j)\in E^\mathrm{or}$, $f_l(i)=\sigma_{ij}f_l(j)$ for all $l=1,\ldots,\rho$.
    Hence, we have that
    \[\langle f_l(i),f_k(i)\rangle=\langle \sigma_{ij}f_l(j),\sigma_{ij}f_k(j)\rangle=\langle f_l(j),f_k(j)\rangle.\]
    Hence for any $i\in V$, $[f_1(i),\ldots,f_\rho(i)]$ is a orthonormal basis for a $\rho$-dimensional subspace of $\R^d$. Then, for all $i\in V$, we expand this basis to an orthonormal basis $f(i):=[f_1(i),\ldots,f_\rho(i),g_{\rho+1}(i),\ldots,g_d(i)]$ for $\R^d$.  This provides a map $f:V\to\R^{d\times d}$.

    We now define a signature $\tau$ as follows: for any oriented edge $(i,j)\in E^\mathrm{or}$,
    \[\tau_{ij}:=[g_{\rho+1}(i),\ldots,g_d(i)]^\mathrm{T}\sigma_{ij}[g_{\rho+1}(j),\ldots,g_d(j)].\]
    It is then straightforward to verify that $\sigma\simeq (\bigoplus_{i=1}^\rho\iota^1)\oplus \tau$ via the switching map $f$. Using \Cref{prop:equivalence connection Laplacian}, $\mathcal{L}^\sigma$ and $ \mathcal{L}^{\oplus_{i=1}^\rho\iota^1}\oplus \mathcal{L}^\tau$ are similar matrices, and hence since the kernel of $\mathcal{L}^{\oplus_{i=1}^\rho\iota^1}$ is exactly $\rho$, $\mathcal{L}^\tau$ is in turn nonsingular. The claim follows.
\end{proof}

Given a signature $\sigma$, if $\mathcal{L}^\sigma$ is invertible, then we call $\sigma$ \emph{absolutely inconsistent}. For $\tau$ in the above theorem, we call $\tau$ the \emph{absolutely inconsistent component} of $\sigma$. 

As an application of the theory developed in this section, we provide a complete characterization of signatures on cycle graphs.

\begin{example}[Elementary cycle signatures]\label{ex:elementary-cycle-sig}
    Consider an $n$-cycle graph $G$. Given $\theta \in\mathbb{R}$, define a 2-dimensional signature $\sigma^\theta$ as follows: let $\sigma^\theta_{12}\in \mathsf{O}(2)$ be the rotation with angle $\theta$, i.e.,
        $\sigma^\theta_{12} = \begin{bmatrix}
            \cos{\theta} & -\sin{\theta}\\
            \sin{\theta} & \cos{\theta}
        \end{bmatrix};$
    let $\sigma^\theta_{i,i+1}=I_{2\times 2}$ for all $i = 2,\ldots,n$. Then, the following observations hold:
        \begin{enumerate}
            \item If $\theta=0\mod{[0,2\pi)}$, $\sigma^\theta\cong \iota^1\oplus\iota^1$ and is in particular consistent;
            \item If $\theta=\pi\mod{[0,2\pi)}$, $\sigma^\theta\cong \iota^{-1}\oplus\iota^{-1}$ where $\iota^{-1}$ denotes the one-dimensional signature that has the value $-1$ at $(1,2)$ and $+1$ elsewhere, and is in particular inconsistent and decomposable;
            \item If $\theta \neq 0,\pi \mod{[0,2\pi)}$ then $\sigma^\theta$ is absolutely inconsistent and indecomposable, since $\sigma^\theta_{12}$ is not diagonalizable with real coefficients.
        \end{enumerate}
        These signatures are called \emph{elementary cycle signatures}.
\end{example}

We show that any signature on a cycle graph is a direct sum of signatures described by \Cref{ex:elementary-cycle-sig}.

\begin{proposition}\label{prop:cycle}
    Consider an $n$-cycle graph $G$ with a $d$-dim signature $\sigma$. Then, there exist $d_1,d_{-1}\in\mathbb{N}$ and $\theta_1,\ldots,\theta_k\in(0,\pi)\cup(\pi,2\pi)$ such that $d_1+d_{-1}+2k=d$ and
    \[\sigma\simeq (\bigoplus_{i=1}^{d_1}\iota^1)\oplus (\bigoplus_{i=1}^{d_{-1}}\iota^{-1})\oplus \sigma^{\theta_1}\oplus\cdots \sigma^{\theta_k}.\]
\end{proposition}

\begin{proof}
    We let $T$ be the spanning tree of $G$ not containing edge $\{1,2\}$.
    By \Cref{lemma:spanning tree}, we have that $\sigma\simeq \sigma^T$. Let $\tau:=\sigma^T_{12}$. Then there exists an orthonormal matrix $P$ such that $\zeta:=P^{-1}\tau P$ is a diagonal block matrix with $d_1$ blocks with value $1$,  $d_{-1}$ blocks with value $-1$, and $k$ 2-dimensional rotation matrix with angles $\theta_1,\ldots,\theta_k$.
    Define a new signature $\sigma^{T,P}$ so that $\sigma^{T,P}_{ij} = P^{-1}\sigma^T P$. Then, $\sigma^{T,P}_{12}=\zeta$ and $I_{d\times d}$ otherwise. Then,
    \[\sigma\simeq \sigma^{T,P}= (\bigoplus_{i=1}^{d_1}\iota^1)\oplus (\bigoplus_{i=1}^{d_{-1}}\iota^{-1})\oplus \sigma^{\theta_1}\oplus\cdots \sigma^{\theta_k}.\]
\end{proof}

%%%%%%%%%%%%%%%%%%%
%%%%%%%%%%%%%%%%%%%
%%%%%%%%%%%%%%%%%%% 

\section{Background on Effective Resistance}
\label{subsec:effective-resistance}
In this section, we present a comprehensive overview of effective resistance on classical graphs. We focus on two fundamental perspectives: the energy perspective and the random walk perspective. Understanding these perspectives is crucial as they serve as the primary sources of inspiration for our theoretical advancements in subsequent sections.

Let $G=(V,E,W)$ be a connected graph and let $i,j\in V$ be any two nodes. The \textit{effective resistance} between $i,j$ is given by
\begin{equation}\label{eq:effective-resistance}
    r_{ij}:=(e_i-e_j)^\mathrm{T}L^\dagger (e_i-e_j).
\end{equation}
where $e_i$ denotes the $i$-th standard basis vector in $\mathbb{R}^n$. The \textit{effective conductance} between $i,j$ is defined by $c_{ij} = r_{ij}^{-1}$.

In \cite{chung2014connection}, this notion of effective resistance was generalized to edges $\{i,j\}$ in the connection graph $(G,\sigma)$ as follows.

\begin{definition}\label{def:connection-resistance}
    Let $(G,\sigma)$ be a fixed, connection graph and let $i,j\in V$ be adjacent vertices. The \textit{ effective (connection) resistance} between $i,j$ is defined by
    \begin{align}
        R^\sigma(i,j)\coloneqq \|M^\mathrm{T}_{i,j}\mathcal{L}^\dagger M_{i,j}\|_2
    \end{align}
    where $M_{i,j} = \begin{bmatrix}
            0_{d\times d}, \cdots,
            I_{d\times d},\cdots, -\sigma_{ij},\cdots
            ,0_{d\times d}\end{bmatrix}^\mathrm{T}$ and $\|\cdot\|_2$ denotes the matrix 2-norm.
\end{definition}

It then follows from \Cref{def:connection-resistance} and \Cref{rmk:path-independence} that whenever $(G,\sigma)$ is consistent and connected, $R^\sigma(i,j)$ is defined for any pair of (not necessarily adjacent) vertices $i,j\in V$. Furthermore, in this case, the connection resistance in fact coincides with the graph effective resistance. The following result is a slight generalization of Theorem 4 in \cite{chung2014connection} where only edges were considered.

\begin{theorem}\label{thm:connection-resistance old}
    Let $\sigma$ be a consistent connection and let $i,j\in V$. Then
    $R^\sigma(i,j) = r_{ij}.$
\end{theorem}

This result is not surprising given that consistent signatures are ``trivial'' in the sense of \Cref{ex:consistent signature}. However, there are two main limitations of this generalization of effective resistance when considering inconsistent connection graphs.
\begin{enumerate}
    \item The definition cannot be generalized to define effective resistance between pairs $i,j$ that are not edges.
    \item Even for a fixed edge $\{i,j\}$, $R^\sigma(i,j)$ is not continuous w.r.t. change of signature $\sigma$. See the example below.
\end{enumerate}

\begin{example}[The connection resistance in \cite{chung2014connection} is discontinuous]\label{ex:discontinuous-connection-resistance}
Consider the 3-cycle graph with vertex set $V=\{1,2,3\}$. For any $\theta\in\R$, let $\sigma^\theta$ denote the 2-dimensional elementary cycle signature (cf. \Cref{ex:elementary-cycle-sig}). Then, for any $\theta\in\R$, the connection resistance between nodes $1,2$ is given by $5+4\cos(\theta)$ whenever $\theta\neq 2k\pi$. Note, however, that when $\theta=2k\pi$, the signature $\sigma^\theta$ is consistent and the connection resistance between nodes $1,2$ is given by $2/3$. Hence, the connection resistance is discontinuous at $\theta=2k\pi$. See our GitHub Repository for the Mathematica code for derivation of the connection resistance. 
\end{example}

One of the main goals of this paper is to provide a novel definition of effective resistance for connection graphs which addresses the two limitations above.

\subsection{An Energy Perspective on Effective Resistance}\label{subsubsec:energy-perspective}
Given a graph $G$, the Dirichlet energy of any function $f:V\to\R$ is defined as $E(f) := f^\mathrm{T}Lf$. It turns out that the effective resistance can be expressed as the Dirichlet energy of a particular function.

\begin{theorem}\label{prop:classical-resistance-properties}
    Let $G=(V,E,W)$ be a connected graph and let $i,j\in V$.
    \begin{enumerate}
        \item $r_{ij} = E(f)$ for any $f:V\to \R$ such that $ Lf = e_i - e_j$;
        \item $c_{ij} = \inf\{E(f): f(i) = 1, f(j) = 0\}$.
    \end{enumerate}
\end{theorem}
Proofs of these results can be found in \cite[Theorem 4.2]{jorgensen2008operator}, or \cite[Theorem 4.1]{lovasz1993random}.

Note that the equation involved in (i) above is called the \emph{Poisson problem} (PP):
\begin{equation}\label{eq: poisson problem classical}
    (Lf)(x)=\begin{cases}
        1 & x=i\\
        -1 & x=j\\
        0 & \text{otherwise}
    \end{cases}.
\end{equation}
(PP) has infinitely many solutions as any solution to (PP) plus a constant function generates a new solution.
The definition \Cref{eq:effective-resistance} of effective resistance can be interpreted as the Dirichlet energy of a solution of (PP): one solution of (PP) is $L^\dagger(e_i-e_j)$ and hence $r_{ij}=E(f)=(e_i-e_j)^\mathrm{T}L^\dagger (e_i-e_j)$.

For (ii) in the above theorem, the infimum is achieved by the solution to the following Dirichlet problem:
\begin{equation}\label{eq:dirichlet-problem-usual}
    \begin{cases}
        Lf|_{V\backslash\{i,j\}} = 0\\
        f(i)=1,\,f(j)=0 
    \end{cases}.
\end{equation}
In this case, the solution to \Cref{eq:dirichlet-problem-usual} is unique and we denote it by $V_{i\to j}:V\to \R$ which is called the \emph{voltage function}. It is worth noting that $c_{ij} = (LV_{i\to j})(i)$ and $-c_{ij} = (LV_{i\to j})(j)$, i.e., if without loss of generality we assume $i=1$ and $j=2$, then
\begin{align}\label{eq:LV}
    LV_{i\to j} = \begin{bmatrix}
        c_{ij}\\
        -c_{ij} \\
        \hline 0_{(n-2)\times 1}
    \end{bmatrix}.
\end{align}

In this way, we see that the effective resistance / conductance can be viewed as the Dirichlet energy of the solution of the Poisson problem / Dirichlet problem. Note that the connection resistance defined in \Cref{def:connection-resistance} follows the spirit of the Poisson problem. However, in \cite{sugiyama2023kron}, the Dirichlet problem is the main source of inspiration for the definition of effective conductance for directed graphs. In \Cref{sec:connection conductance and resistance}, we will show how the Dirichlet problem perspective can give rise to a novel definition of connection conductance and eventually connection resistance.

\subsection{Graph Random Walks and Effective Resistance / Conductance}
\label{subsubsec:random-walks}
We now summarize some basic relationships between effective resistance / conductance and random walks on graphs. For a more detailed discussion, see \cite{lovasz1993random} and \cite{doyle1984random}.

Given a graph $G=(V,E,W)$, we let $(X_t)_{t\geq 0}$ be a simple random walk with transition kernel $D^{-1}W$. For any vertices $i,j$, the transition probability is denoted as $\mathbb{P}_{i,j}=w_{ij}/\deg(i)$. We also use the shorthand notation $\mathbb{P}^i[\cdot]$ to denote the conditional probability $\mathbb{P}[\cdot|X_0=i]$. Similarly, we use the shorthand notation $\mathbb{E}^i[\cdot]$ to denote the conditional expectation $\mathbb{E}[\cdot|X_0=i]$.

We also define various types of stopping times for this random walk as follows.

\begin{definition}
    For $s=0,1$ and any subset $A\subseteq V$, we define 
    \begin{align}\label{def:set-stopping-time}
        T_{A}^s = \inf\{t\geq s: X_t \in A \}.
    \end{align}
    When $A=\{i\}$, we use the shorthand notation $T_i^s$ for $T_{\{i\}}^s$.
\end{definition}
It is useful to note that as long as $G$ is connected, $\mathbb{P}^i[T_j^s <\infty] = 1$ for any $i,j\in V$ and $s$. Alternatively, $\mathbb{P}^i[T_j^s = \infty] = 1$ if and only if $i,j$ are in different connected components.

Then, we introduce the notion of \emph{commute time} and \emph{escape probability}.
We let $H_{ij}:=\mathbb{E}^iT^0_j$ denote the expected number of steps for a simple random walk to reach $j$, having started at $i$. The sum $H_{ij}+H_{ji}$ is called the \emph{commute time} between $i,j$.

The quantity $\mathbb{P}^i[T^1_j<T^1_i]$ denotes the probability that a simple random walk on $G$, having started at node $i$, reaches node $j$ before returning to node $i$. Hence, $\mathbb{P}^i[T^1_j<T^1_i]$ is also called the \emph{escape probability}. 

Then, we have the following results characterizing the effective resistance and conductance in terms of random walks.
\begin{theorem}\label{thm:effective-resistance-random-walk}
    Let $G=(V,E,W)$ be a connected graph and let $i,j\in V$.
    \begin{enumerate}
        \item $r_{ij} = \frac{1}{\operatorname{vol}(G)}(H_{ij}+H_{ji})$, where $\operatorname{vol}(G) := \sum_{e\in E}w_e$;
        \item $c_{ij} = \deg(i) \mathbb{P}^i[T^1_j<T^1_i]$.
    \end{enumerate}  
\end{theorem}

We note that (ii) will be generalized to the case of connection graphs in our \Cref{thm:escape probability} whereas the generalization of (i) remains an open problem.
%%%%%%%%%%%%%%%%%%%%

\section{Dirichlet Problems and Random Walks on Connection Graphs}
\label{sec:random walks}
Building upon the perspectives of Dirichlet problems and random walks for studying effective resistance in classical graphs (cf. \Cref{subsubsec:energy-perspective} and \Cref{subsubsec:random-walks}), in this section, we explore their relevance to connection graphs. These concepts and findings will contribute to our formulation of effective conductance and resistance for connection graphs in the subsequent section.

\subsection{Dirichlet Problems and Harmonic Functions on Connection Graphs}
\label{subsec:harmonic-functions-connection}

In this subsection we define Dirichlet problems and harmonic functions for connection graphs. We establish uniqueness of solutions to such Dirichlet problems under certain conditions and provide an energy characterization of such solutions.

Specifically, when dealing with a graph $G$, we shift our focus from vector-valued functions as discussed in \Cref{subsec:connection-graph-basics} to matrix-valued functions $f\in\ell^2(V;\mathbb{R}^{d\times d})$. As we proceed, it becomes evident that this space is most suitable for the study of harmonic functions for our intended purposes, interpreting the effective conductance as analogous to ``hitting time", which now is related to those expected mean path signatures that are represented as matrices.

Worth noting is that whenever $d=1$ and the connection $\sigma\equiv 1$ is trivial, the results we develop match those for harmonic functions on classical graphs.

Let $(G,\sigma)$ be a connection graph. Suppose $H\subset V$.  Then $f$ is said to be \textit{harmonic} on $H$ if $(\mathcal{L}f)(i) = 0_{d\times d}$ for each $i\in H$. A useful observation is that if $f$ is harmonic on $H$, then for each $i\in H$, $f$ satisfies the mean value property:
\[f(i) = \frac{1}{\deg(i)}\sum_{j\sim i}w_{ij}\sigma_{ij}f(j).\]

We define the vertex boundary of $H$ by
\begin{align}
    \partial H = \left\{j\in V:j\notin H, \text{ and there exists some }i\in H\text{ such that }i\sim j\right\}
\end{align}
We define the vertex closure of $H$ by $\overline{H} = H\cup\partial H$.

\begin{proposition}[Maximum norm principle]\label{prop:maximum-norm-principle}
    Let $(G,\sigma)$ be a connection graph, and let $H\subsetneq V$ be a nonempty proper subset of vertices in $G$. Suppose $f\in \ell_2(V;\mathbb{R}^{d\times d})$ is harmonic on $H$. Let $\|\cdot\|$ be any orthogonally invariant matrix norm on $\mathbb{R}^{d\times d}$. Then
    \begin{align}
        \max_{i\in \overline{H}} \|f(i)\| = \max_{i\in \partial H} \|f(i)\|.
    \end{align}
\end{proposition}

\begin{proof}
    Suppose there exists some $i^*\in H$ for which $\|f(i^\ast)\|\geq \|f(i)\|$ for each $i\in\overline{H}$. Then by the mean value property,
    \begin{align}
        \|f(i^*)\| =\frac{1}{\deg(i)}\left\|\sum_{j\sim i^*}w_{i^*j}\sigma_{i^* j}f(j)\right\|\leq \frac{1}{\deg(i)}\sum_{j\sim i^*}w_{i^*j}\left\|\sigma_{i^*j}f(j)\right\|\leq \|f(i^\ast)\|.
    \end{align}
    Therefore, $\|f(j)\| = \|f(i^\ast)\|$ for each $j\sim i^*$. By iterating this argument, it follows that $\|f(\cdot)\|$ is constant on the set $S\subset \overline{H}$ of nodes which are reachable along a path starting at $i^*$ contained strictly in $\overline{H}$. Since $H$ is a proper subset of $V$ and hence must have a nonempty boundary, $S$ contains at least one boundary node and hence $\|f(i^*)\|$ is achieved on the boundary.
\end{proof}

Given any function $\phi:\partial H\rightarrow \mathbb{R}^{d\times d}$, the \emph{Dirichlet problem} (DP) is given by
\begin{equation}\label{eq:dirichlet-problem}
    \begin{cases}
        u|_{\partial H} = \phi            \\
        \mathcal{L}u|_{H} = 0_{d\times d}
    \end{cases},\,\text{ where } u\in\ell^2(V;\mathbb{R}^{d\times d}).
\end{equation}

\begin{corollary}\label{corol:uniqueness-of-DP}
    If $H\subsetneq V$ satisfies the conditions outlined in the statement of  \Cref{prop:maximum-norm-principle}, the solution to (DP) is unique for any choice of $\phi$ (since, in particular, the difference of any two solutions $f_1-f_2$ is also harmonic and has norm zero on the boundary). In particular, if we choose $H= V\backslash\{i\} =:i^c$ where $i\in V$ is a single fixed node, it follows that the submatrix $\mathcal{L}_{i^c, i^c}$ is positive definite.
\end{corollary}

Finally, we provide a useful characterization of the solution of (DP) in terms of the Dirichlet energy.
For any given $f\in \ell^2(V;\mathbb{R}^{d\times d})$ we define the \textit{(connection) Dirichlet energy} of $f$, analogous to the vector case in \Cref{eq:connection-dirichlet-energy-vector-case}, by the equation
\begin{align}\label{eq:connection-dirichlet-energy}
    E(f) :=\frac{1}{2}\tr\left(f^\mathrm{T} \mathcal{L} f\right) = \frac{1}{2}\tr\left(\sum_{\{i,j\}\in E}w_{ij}(f(i) - \sigma_{ij}f(j))^\mathrm{T}(f(i)-\sigma_{ij}f(j))\right)
\end{align}

\begin{proposition}[Energy minimization]
    \label{prop:energy-min}
    Suppose $(G,\sigma)$ is a connection graph and $H\subsetneq V$ is a proper subset of vertices of $G$. Assume without loss of generality that $V = \overline{H}$ and fix $\phi:\partial H\rightarrow\mathbb{R}^{d\times d}$. Then $f_0 \in\ell^2(V;\mathbb{R}^{d\times d})$ is a solution to (DP) in  \Cref{eq:dirichlet-problem} if and only if $f_0$ is also a solution to the corresponding energy minimization problem (EP), i.e.,
    \begin{align}
        E(f_0) = \min_{f|_{\partial H} = \phi} E(f).
    \end{align}
\end{proposition}

A proof of this proposition can be found in \Cref{sec:missing_proofs}.

\subsection{Graph Random Walks and Mean Path Signatures}
\label{subsec:vector-random-walks}

Consider a simple random walk $(X_t)_{t\geq 0}$  with transition kernel $D^{-1}W$ on a graph $G=(V,E,W)$. If $G$ is equipped with a signature $\sigma$, along the random walk $(X_t)_{t\geq 0}$, one can also record the signature encountered. We hence define the notion of \emph{mean path signature}. 

\begin{definition}[Mean path signature]\label{defn:mean-path-product}
    Let $(G,\sigma)$ be a connection graph and let $i,j,k\in V$. We define
        \[\Omega_{ij}^s := \mathbb{E}^i\left[\prod_{\ell=1}^{T_{j}^s}\sigma_{X_{\ell-1}X_{\ell}}\right] \text{ and }\Omega_{ij}^s(k) := \mathbb{E}^i\left[\prod_{\ell=1}^{T_{j}^s}\sigma_{X_{\ell-1}X_{\ell}}\middle| T_j^s< T_k^s \right].\]
\end{definition}

In words, $\Omega_{ij}^s$ is the expected mean path signature for a walk starting at $i$ and terminating at $j$ whereas $\Omega_{ij}^s(k)$ is the expected mean path signature for a walk starting at $i$ and terminating at $j$ before hitting $k$. As a matter of convention, we let $\Omega_i^s:=\Omega_{ii}^s$. Similarly, we define $\Omega_{i}^s(k) :=\Omega_{ii}^s(k)$.

The choice of $s=0$ or $1$ allows this definition to be toggled as needed in the case $i=j$. If $s=0$ then $\Omega_{i}^0 = I_{d\times d}$ by default, and if $s=1$, $\Omega_{i}^s$ can be interpreted as a mean cycle product over the cycles that begin and end at node $i$. 
Lastly, it is also important to observe that while $\sigma:E^{\text{or}}\rightarrow\mathsf{O}(d)$, $\Omega_{ij}^s$ need not be orthogonal; rather, it is a convex combination of rotation matrices located somewhere within the convex hull of $\mathsf{O}(d)$.

\begin{example}\label{ex:mean-path-products}(Consistent graphs and cycle graphs)
    \begin{enumerate}
        \item {\normalfont (Consistent Graphs)} If $(G,\sigma)$ is consistent, then by \Cref{rmk:path-independence}, for any three nodes $i,j,k\in V$, and $s=0,1$ we have that 
              $\Omega^s_{ij} = \Omega^s_{ij}(k) = \sigma_{ij}.$
              In particular, when $\sigma_{ij} = I_{d\times d}$ for any $i,j$, $\Omega^s_{ij} = \Omega^s_{ij}(k) = I_{d\times d}$.

        \item {\normalfont (Cycle Graph)} In the case where $G$ is a cycle graph on $n$ vertices, $i,j\in V$ with $i\neq j$, and $\sigma$ is any signature, we have by inspection that
            $\Omega_{i}^1(j)  = I_{d\times d}.$

    \end{enumerate}
\end{example}

It is natural to wonder whether there is any relationship between $\Omega_{ij}^s$ and $\Omega_{ji}^s$ (resp. $\Omega_{ij}^s(k)$ and $\Omega_{ji}^s(k)$). In general, we do not have that $\Omega_{ij}^s(k) = (\Omega_{ji}^s(k))^\mathrm{T}$ or $\Omega_{ij}^s = (\Omega_{ji}^s)^\mathrm{T}$ (see \Cref{fig: path} for an illustration). But the following result provides a positive answer to some extent.

\begin{figure}[htb!]
    \centering
    \includegraphics[width = 0.2\linewidth]{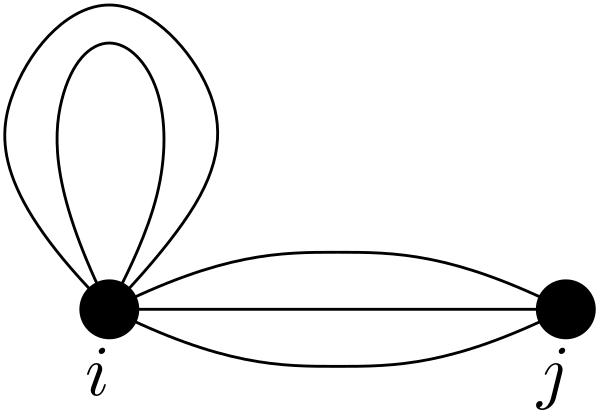}
    \caption{A path contributing to $\Omega_{ij}^s$ (or $\Omega_{ij}^s(k)$ when $k\neq i,j$) can involve multiple loops passing through $i$, which is not allowed for a path contributing to $\Omega_{ji}^s$ (or $\Omega_{ji}^s(k)$ when $k\neq i,j$). This inherent asymmetry is the reason why, in general, $\Omega_{ij}^s(k) \neq (\Omega_{ji}^s(k))^\mathrm{T}$ or $\Omega_{ij}^s \neq (\Omega_{ji}^s)^\mathrm{T}$.}
    \label{fig: path}
\end{figure}

\begin{proposition}\label{prop:mean-path-product-transpose}
    Let $(G,\sigma)$ be a connection graph and let $i,j,k\in V$. Then, for any $s=0,1$, one has that
    $\Omega_{ij}^s(i)  = (\Omega_{ji}^s(j))^\mathrm{T}.$
\end{proposition}
The proof can be found in \Cref{sec:missing_proofs}.

Note that the mean path signatures can be used to characterize the kernel of $\mathcal{L}$.

% \sawyer{Adjust the equations typesetting to save room.}
\begin{proposition}[A characterization of kernel of $\mathcal{L}$]
    Let $f:V\rightarrow\R^{d\times d}$ be such that $(\mathcal{L}f)(x)= 0_{d\times d}$ for all $x\in V$. Then, for any $i,j\in V$, one has that
    $f(i) = \Omega_{ij}^0f(j).$
\end{proposition}

\begin{proof}
    Notice that for any edge $\{i,j\}$, $f$ satisifes that $f(i) = \sigma_{ij}f(j)$. Then, for any $i,j\in V$, one has that
    \begin{align*}
        \Omega_{ij}^0f(j) & = \mathbb{E}^i\left[\prod_{\ell = 1}^{T_j^0}\sigma_{X_{\ell-1} X_{\ell}}\cdot f(j)\right]  = \mathbb{E}^i\left[\prod_{\ell = 1}^{T_j^0-1}\sigma_{X_{\ell-1} X_{\ell}}\cdot \sigma_{X_{T_j^0-1} j} f(j)\right] \\
        & = \mathbb{E}^i\left[\prod_{\ell = 1}^{T_j^0-1}\sigma_{X_{\ell-1} X_{\ell}}\cdot f(X_{T_j^0-1})\right]=\cdots = \mathbb{E}^i\left[\sigma_{i X_{1}}\cdot f(X_{1})\right]  = f(i).
    \end{align*}
\end{proof}

\paragraph{Mean path signatures vs Dirichlet problems} The mean path signature $\Omega_{ij}^0$ turns out to be closely related to Dirichlet problems on connection graphs. 
Let $(G,\sigma)$ be a connection graph and let $j\in V$ be a fixed node. Consider the map $\Omega_{\bullet j}^0:V\to \R^{d\times d}$ sending each $i\in V$ to $\Omega_{ij}^0$. Using the Markov property and the uniqueness of the solution to the Dirichlet problem (cf. \Cref{corol:uniqueness-of-DP}), we have the following lemma.

\begin{lemma}\label{lem:Omega-dirichlet-conditions}
    The map $\Omega_{\bullet j}^0$ is the unique solution to the Dirichlet problem below:
    \begin{align}\label{eq:Omega-dirichlet-conditions-1}\begin{cases}
        \Omega_{jj}^0=I_{d\times d}                                \\
        \mathcal{L}\Omega_{\bullet j}^0|_{V\backslash \{j\}} = 0_{d\times d}
        \end{cases}.\end{align}
\end{lemma}

The above lemma indicates the following explicit way of calculating $\Omega_{ij}^0$ for $i,j\in V$ using the connection Laplacian.

\begin{proposition}[Calculation of $\Omega_{ij}^0$]\label{prop:cal of omega}
    Let $(G,\sigma)$ be a connection graph and let $i,j\in V$ be fixed nodes with $i\neq j$. Write $\mathcal{L}$ as a $2\times 2$ block matrix of the form (re-order the nodes of $G$ if necessary) $\mathcal{L} = \left[\begin{array}{c|c}
        \mathcal{L}_j & \mathcal{L}_{j, j^c} \\
        \hline
        \mathcal{L}_{j^c, j} & \mathcal{L}_{j^c}
    \end{array}\right]
    ,$
    where $j^c:=V\backslash\{j\}$. Then, it follows:
        \begin{align}
            \Omega_{ij}^0 = -\left((\mathcal{L}_{j^c, j^c})^{-1}\mathcal{L}_{j^c, j}\right)(i)
        \end{align}
    where by $(\cdot)(i)$ we mean the $d\times d$ block matrix component of $(\cdot)$ corresponding to the node $i$.
    
\end{proposition}
\begin{proof}
    By \Cref{lem:Omega-dirichlet-conditions}, $\Omega_{ij}^0$ is the unique solution to the Dirichlet problem \Cref{eq:Omega-dirichlet-conditions-1}.
    We partition $\Omega_{\bullet j}^0$ (regarded as a block matrix) and $\mathcal{L}$ to re-write \Cref{eq:Omega-dirichlet-conditions-1} in block matrix form:
        \begin{align}\label{eq:block-matrix-dirichlet}
            \left[\begin{array}{c|c}
                \mathcal{L}_j & \mathcal{L}_{j, j^c} \\
                \hline
                \mathcal{L}_{j^c, j} & \mathcal{L}_{j^c} 
            \end{array}\right]
\begin{bmatrix}
                I_{d\times d}\\\hline \Omega_{\bullet j}^0|_{j^c}                
            \end{bmatrix} = \begin{bmatrix}
                (\mathcal{L}\Omega_{\bullet j}^0)(j) \\\hline 0_{(n-1)d\times d}
            \end{bmatrix}.
        \end{align}
    Focusing on the lower term of the right hand side, we can carry out block matrix multiplication to write:
        \begin{align}
            \mathcal{L}_{j^c, j} I_{d\times d} + \mathcal{L}_{j^c} \Omega_{\bullet j}^0|_{j^c} = 0_{(n-1)d\times d}.
        \end{align}
    Equivalently, $\mathcal{L}_{j^c}\Omega_{\bullet j}^0|_{j^c} = -\mathcal{L}_{j^c, j}$. Using \Cref{corol:uniqueness-of-DP}, we know that $\mathcal{L}_{j^c}$ is positive definite and hence $ \Omega_{ij}^0 = -((\mathcal{L}_{j^c})^{-1}\mathcal{L}_{j^c, j})(i)$ for any $i\neq j$. This concludes the proof.
\end{proof}

A useful observation from \Cref{eq:block-matrix-dirichlet} is that
    \begin{equation}\label{rmk:schur j}
        \mathcal{L}/\mathcal{L}_{j^c}=\deg(j)\Omega_{jj}^0-\sum_{x}A_{jx}\sigma_{jx}\Omega_{xj}^0=\deg(j)(I_{d\times d} - \Omega_{jj}^1).
    \end{equation}

\subsection{Mean Path Signatures under Equivalence and Direct Sum}\label{subsec:mean path signatures equivalence}
Let $G$ be a graph and let $\sigma$ be a signature on $G$.
To specify the given signature, we let $\Omega_{ij}^{\sigma,s}$ and $\Omega_{ij}^{\sigma,s}(k)$ denote the mean path signatures with respect to $\sigma$.
The next proposition shows that the mean path signatures are essentially ``invariant'' under switching equivalence.

\begin{proposition}\label{prop:mean path equivalence}
    Assume that $\sigma\simeq\tau$ and  let $f:V\to\mathsf{O}(d)$ be a switching map.
    Then for any $i,j,k\in V$, one has that
    \[f(i)\Omega_{ij}^{\sigma,s}=\Omega_{ij}^{\tau,s}f(j)\text{ and }f(i)\Omega_{ij}^{\sigma,s}(k)=\Omega_{ij}^{\tau,s}(k)f(j).\]
\end{proposition}

Now, we consider mean path signatures for direct sums of signatures.

\begin{proposition}\label{prop:mean path direct sum}
    Let $\sigma$ and $\tau$ be signatures on $G$. Then for any $i,j,k\in V$, one has
    \[\Omega_{ij}^{\sigma\oplus \tau,s}=\Omega_{ij}^{\sigma,s}\oplus\Omega_{ij}^{\tau,s}\,\text{ and }\,\Omega_{ij}^{\sigma\oplus \tau,s}(k)=\Omega_{ij}^{\sigma,s}(k)\oplus\Omega_{ij}^{\tau,s}(k).\]
\end{proposition}

The following result provides another characterization of absolutely inconsistent signatures. The proof can be found in \Cref{sec:missing_proofs}.

\begin{theorem}\label{thm:absolutely inconsistent via mean path signature}
    Let $(G,\sigma)$ be a connection graph. Then, $\sigma$ is absolutely inconsistent iff for all $i\in V$ one has that $I_{d\times d}-\Omega_i^1$ is \emph{positive definite}.
\end{theorem}

%%%%%%%%%%%%%%%%%%%%%%%%%%%%%%%
\section{Conductance on Connection Graphs}
\label{sec:connection conductance and resistance}

In this section, we utilize the Dirichlet problem and mean path signatures established in \Cref{sec:random walks} to develop the effective conductance in connection graphs. Just as we consider matrix valued functions in the Dirichlet problem for connection graphs, we establish the framework of effective conductance \textit{matrices} for connection graphs, by analogy to the classical setting.

Inspired by the relationship between the effective conductance and the Dirichlet problem in \Cref{subsubsec:energy-perspective}, we introduce the notion of the connection conductance matrix with respect to two nodes in $G$ by utilizing a Dirichlet problem on connection graphs. Note that this approach has also been implicitly adopted by \cite{sugiyama2023kron} to define effective conductance for directed graphs and by \cite{song2019extension} to define effective conductance between disjoint sets of vertices. 

Now, given $(G,\sigma)$, consider the following Dirichlet problem for any $i,j\in V$:
\begin{equation}\label{eq: dirichlet boundary problem connection graph}
    \begin{cases}
        \mathcal{L}^\sigma\mathcal{V}_{i\to j}|_{V\backslash \{i,j\}} = 0\\
        \mathcal{V}_{i\to j}(i)=I_{d\times d}, \mathcal{V}_{i\to j}(j)=0_{d\times d}
        \end{cases}
\end{equation}
By \Cref{corol:uniqueness-of-DP}, the above equation has a unique solution $\mathcal{V}_{i\to j}:V\to \R^{d\times d}$ which we call the \emph{connection voltage function} from $i$ to $j$. We in particular care about $(\mathcal{L}^\sigma\mathcal{V}_{i\to j})(i)$ and $(\mathcal{L}^\sigma\mathcal{V}_{i\to j})(j)$ as in the usual graph case, one can recover effective conductance from either term (cf. \Cref{eq:LV}). After possibly re-enumerating the nodes of $G$ as needed, we see that
    \begin{align}\label{eq:laplacian-of-voltage}
        \mathcal{L} \mathcal{V}_{i\to j} = \begin{bmatrix}
            (\mathcal{L}\mathcal{V}_{i\to j})(i) \\
            (\mathcal{L}\mathcal{V}_{i\to j})(j) \\
            \hline 0_{(n-2)d\times d}
        \end{bmatrix}.
    \end{align}
We can expand the preceding equation by partitioning the connection Laplacian $\mathcal{L}$ according to the nodes $i,j$, in the style of the proof of \Cref{prop:cal of omega}, as follows:
\[\left[\begin{array}{c|c}
    \mathcal{L}_{\{i,j\},\{i,j\}}   & \mathcal{L}_{\{i,j\},\{i,j\}^c}   \\
    \hline
    \mathcal{L}_{\{i,j\}^c,\{i,j\}} & \mathcal{L}_{\{i,j\}^c,\{i,j\}^c} \\
\end{array}\right]
\begin{bmatrix}
        I_{d\times d}  \\
        0_{d\times d}  \\
        \hline
        \mathcal{V}_{i\to j}|_{\{i,j\}^c}
    \end{bmatrix} = \begin{bmatrix}
        (\mathcal{L}\mathcal{V}_{i\to j})(i) \\
        (\mathcal{L}\mathcal{V}_{i\to j})(j) \\
        0_{(n-2)d\times d}
    \end{bmatrix}.\]
Then it follows from a straightforward block matrix calculation that
\[
    \left(\mathcal{L}_{\{i,j\}} - \mathcal{L}_{\{i,j\},\{i,j\}^c}            \mathcal{L}_{\{i,j\}^c}^{-1}\mathcal{L}_{\{i,j\}^c,\{i,j\}} \right)\begin{bmatrix}
        I_{d\times d} \\
        0_{d\times d}
    \end{bmatrix} = \mathcal{L}/\mathcal{L}_{\{i,j\}^c} \begin{bmatrix}
        I_{d\times d} \\
        0_{d\times d}
    \end{bmatrix} = \begin{bmatrix} (\mathcal{L}\mathcal{V}_{i\to j})(i) \\
        (\mathcal{L}\mathcal{V}_{i\to j})(j)
    \end{bmatrix}.
\]
Note that the invertibility of $\mathcal{L}_{\{i,j\}^c}$ follows from \Cref{corol:uniqueness-of-DP}.

Based on the discussion above, we now define the connection conductance matrix.

\begin{definition}[Connection Conductance Matrix]\label{defn:connection-conductance-matrix}
    Let $(G,\sigma)$ be a connection graph and let $i,j\in V$ be any two nodes. Then the connection conductance matrix $\mathcal{C}^\sigma(i,j)\in\mathbb{R}^{2d\times 2d}$ is given by the Schur complement of $\mathcal{L}$ with respect to $\{i,j\}^c$, as follows:
        \[\mathcal{C}^\sigma(i,j) = \mathcal{L}/\mathcal{L}_{\{i,j\}^c}.\]
    where $\mathcal{L}/\mathcal{L}_{\{i,j\}^c} = \mathcal{L}/\mathcal{L}_{ \{i,j\}^c, \{i,j\}^c}$ by convention. For later use, and with a small abuse of notation, we explicitly write the $d\times d$ blocks of $\mathcal{C}^\sigma(i,j)$ in the form
        $\mathcal{C}^\sigma(i,j) = \begin{bmatrix}
            \mathcal{C}_{ii}^\sigma & \mathcal{C}_{ij}^\sigma \\
            \mathcal{C}_{ji}^\sigma & \mathcal{C}_{jj}^\sigma\end{bmatrix}.$
\end{definition}

Note that from the calculation preceding \Cref{defn:connection-conductance-matrix}, it follows that
    \begin{align}\label{eq:conductance-laplacian-voltage}\mathcal{C}^\sigma(i,j) = \mathcal{L}/\mathcal{L}_{\{i,j\}^c} = \left[\begin{array}{c|c}
        (\mathcal{L}\mathcal{V}_{i\rightarrow j})(i) & (\mathcal{L}\mathcal{V}_{j\rightarrow i})(i) \\
        \hline
        (\mathcal{L}\mathcal{V}_{i\rightarrow j})(j) & (\mathcal{L}\mathcal{V}_{j\rightarrow i})(j)
    \end{array}\right].\end{align}

\begin{remark}[Conductance matrices for classical graphs]\label{rmk:trivial-conductance-matrix}
     If $(G,\sigma)$ is a connection graph with the trivial one-dimensional signature $\sigma=\iota^1$, then $\mathcal{L} = L$ and for any fixed $i,j\in V$, one has that
        $\mathcal{C}^\sigma(i,j) = \begin{bmatrix}
            c_{ij}  & -c_{ij} \\
            -c_{ij} & c_{ij}
        \end{bmatrix}$,
    where $c_{ij}$ is the effective conductance between $i,j$ (as defined in \Cref{subsec:effective-resistance}). This is a well-known result; see, e.g., \cite{schild2018schur}.
\end{remark}

Following the continuity of Schur complement, we note that the conductance matrix is continuous with respect to change of signatures. 
\begin{proposition}\label{prop: continuity of conductance matrix}
    When the graph $G$ is connected, the conductance matrix $\mathcal{C}^\sigma(i,j)$ is continuous with respect to change of signatures $\sigma$.
\end{proposition}

\paragraph{Conductance matrix vs mean path signatures} Before moving on to more in-depth discussions of its properties, we establish some straightforward relationships between blocks in $\mathcal{C}^\sigma(i,j)$ and the mean path signatures. We first note that $C_{ii}^\sigma$ is invertible from \Cref{lm:invertible Cii} which we will introduce later.

\begin{lemma}\label{prop: mean path signature vs schur}
    For any $i\neq j\in V$, one has that
    $\Omega_{ij}^0 = -(\mathcal{C}_{ii}^\sigma)^{-1}\mathcal{C}_{ij}^\sigma.$
\end{lemma}

\begin{proof}
    Without loss of generality, we assume that $i=1$ and $j=2$.
    Then, by \Cref{lem:Omega-dirichlet-conditions}, we have that $\Omega_{ij}^0$ satisfies the following equation:
    \[\left[\begin{array}{c|c}
        \mathcal{L}_{\{1,2\},\{1,2\}}   & \mathcal{L}_{\{1,2\},\{1,2\}^c}   \\
        \hline
        \mathcal{L}_{\{1,2\}^c,\{1,2\}} & \mathcal{L}_{\{1,2\}^c,\{1,2\}^c} \\
    \end{array}\right]
     \begin{bmatrix}
            \Omega_{12}^0 \\
            I_{d\times d} \\
            \hline
            \Omega_{32}^0 \\
            \vdots        \\
            \Omega_{n2}^0
        \end{bmatrix} = \begin{bmatrix}
            0_{d\times d}                        \\
            (\mathcal{L}\Omega_{\bullet,i}^0)(2) \\
            \hline 0_{d\times d}                 \\
            \vdots                               \\
            0_{d\times d}
        \end{bmatrix}.\]
    Therefore, one has that
    \[\mathcal{L}/\mathcal{L}_{\{1,2\}^c}\begin{bmatrix}
            \Omega_{ij}^0 \\
            I_{d\times d}
        \end{bmatrix}=\begin{bmatrix}
            0_{d\times d}                        \\
            (\mathcal{L}\Omega_{\bullet,i}^0)(2)
        \end{bmatrix},\text{ and hence }
    \begin{bmatrix}
            \mathcal{C}_{ii}^\sigma & \mathcal{C}_{ij}^\sigma \\
            \mathcal{C}_{ji}^\sigma & \mathcal{C}_{jj}^\sigma\end{bmatrix}\begin{bmatrix}
            \Omega_{ij}^0 \\
            I_{d\times d}
        \end{bmatrix}=\begin{bmatrix}
            0_{d\times d}                        \\
            (\mathcal{L}\Omega_{\bullet,i}^0)(2)
        \end{bmatrix}.\]
    This implies that
    $\mathcal{C}_{ii}^\sigma\Omega_{ij}^0+\mathcal{C}_{ij}^\sigma=0_{d\times d}$
    and thus we conclude the proof.
\end{proof}

It turns out that the Schur complements of $\mathcal{C}^\sigma(i,j)$ are related to the mean path signature as well.

\begin{lemma}\label{lm: schur complement vs mean path signature}
    For the conductance matrix $\mathcal{C}^\sigma(i,j) = \begin{bmatrix}
            \mathcal{C}_{ii}^\sigma & \mathcal{C}_{ij}^\sigma \\
            \mathcal{C}_{ji}^\sigma & \mathcal{C}_{jj}^\sigma\end{bmatrix}$, one has the following results regarding Schur complements:
    \[\mathcal{C}^\sigma(i,j)/\mathcal{C}_{jj}^\sigma = \deg(i)(I_{d\times d} - \Omega_{i}^1) \quad\text{ and }\quad\mathcal{C}^\sigma(i,j)/\mathcal{C}_{ii}^\sigma = \deg(j)(I_{d\times d} - \Omega_{j}^1).\]
\end{lemma}

\begin{proof}
    This follows from the quotient formula for Schur complements  (cf. \Cref{eq:quotient}) and \Cref{rmk:schur j}:
    \begin{align*}
        \mathcal{C}^\sigma(i,j)/\mathcal{C}_{jj}^\sigma = \mathcal{L}/\mathcal{L}_{i^c} = \deg(i)(I_{d\times d} - \Omega_{i}^1)
    \end{align*}
\end{proof}

\subsection{Conductance Matrix under Equivalence and Direct Sum}\label{subsec: conductance matrix under equivalence and direct sum}

In this subsection we will establish some properties of the conductance matrix under equivalence and direct sum of signatures. Proofs are elementary and can be found in \Cref{sec:missing_proofs}.

\begin{proposition}[Conductance matrix under equivalence]\label{prop:conductance matrix under equivalence}
    Assume that $\sigma\simeq\tau$ and  let $f:V\to\mathsf{O}(d)$ be a switching map.
    Then for any $i,j\in V$, one has that
    \[F_{ij} \mathcal{C}^\sigma(i,j)=\mathcal{C}^\tau(i,j)F_{ij},\,\text{ where } \,F_{ij}:=\begin{bmatrix}
        f(i)          & 0_{d\times d} \\
        0_{d\times d} & f(j)
    \end{bmatrix}.\]
\end{proposition}

Directly from the result above, we also obtain the following relationships for blocks of the conductance matrix:
\begin{enumerate}
    \item $f(i)\mathcal{C}^\sigma_{ii}=\mathcal{C}^\tau_{ii}f(i)$ and $f(j)\mathcal{C}^\sigma_{jj}=\mathcal{C}^\tau_{jj}f(j)$;
    \item $f(i)\mathcal{C}^\sigma_{ij}=\mathcal{C}^\tau_{ij}f(j)$ and $f(j)\mathcal{C}^\sigma_{ji}=\mathcal{C}^\tau_{ji}f(i)$.
\end{enumerate}

\begin{proposition}[Conductance matrix under direct sum]\label{prop:conductance matrix under direct sum}
    Let $\sigma$ and $\tau$ be two signatures on $G$. Then for any $i,j\in V$, one has that $\mathcal{C}^{\sigma\oplus\tau}(i,j)$ is similar to $\mathcal{C}^\sigma(i,j)\oplus\mathcal{C}^\tau(i,j)$.
\end{proposition}

In the manner of \Cref{thm:decomposition}, given a decomposition of any signature $\sigma$ on $G$: $\sigma\simeq (\bigoplus_{i=1}^\rho\iota^1)\oplus \tau$, where $\tau$ is absolutely inconsistent, as a result of the \Cref{prop:conductance matrix under direct sum} and \Cref{rmk:trivial-conductance-matrix}, one can decompose the conductance matrix as follows:
    \[\mathcal{C}^\sigma(i,j)\simeq \left(\bigoplus_{l=1}^\rho\begin{bmatrix}
            c_{ij}  & -c_{ij} \\
            -c_{ij} & c_{ij}
        \end{bmatrix}\right)\oplus \mathcal{C}^\tau(i,j).\]
In particular, we point out that for the blocks of the conductance matrix, one has that
\begin{enumerate}
    \item $\mathcal{C}^\sigma_{ii}\simeq \bigoplus_{l=1}^\rho\begin{bmatrix}
              c_{ij}
          \end{bmatrix}\oplus \mathcal{C}^\tau_{ii}$ and $\mathcal{C}^\sigma_{jj}\simeq \bigoplus_{l=1}^\rho\begin{bmatrix}
            c_{ij}
        \end{bmatrix}\oplus \mathcal{C}^\tau_{ii}$;
    \item $\mathcal{C}^\sigma_{ij}\simeq \bigoplus_{l=1}^\rho\begin{bmatrix}
              -c_{ij}
          \end{bmatrix}\oplus \mathcal{C}^\tau_{ij}$ and $\mathcal{C}^\sigma_{ji}\simeq \bigoplus_{l=1}^\rho\begin{bmatrix}
              -c_{ij}
          \end{bmatrix}\oplus \mathcal{C}^\tau_{ji}$.
\end{enumerate}

It is then particularly interesting to study the conductance matrix of the absolutely inconsistent signature $\tau$.

\subsection{A Physical Perspective on the Conductance Matrix}\label{subsubsec:physical-perspective-conductance}
Motivated by the physical origin of classical effective resistance/conductance, in this section, we provide a physical interpretation of the Dirichlet problem \Cref{eq: dirichlet boundary problem connection graph} and hence the corresponding conductance matrix.

We write down explicitly \Cref{eq:laplacian-of-voltage} at each vertex $i\in V$ as follows to obtain current-balance equations:
\[\mathcal{C}_{ii}^\sigma = \sum_x A_{ix}(\mathcal{V}_{i\rightarrow j}(i) - \sigma_{ix}\mathcal{V}_{i\rightarrow j}(x)),\]
\[\mathcal{C}_{ji}^\sigma = \sum_x A_{jx}(\mathcal{V}_{i\rightarrow j}(j) - \sigma_{jx}\mathcal{V}_{i\rightarrow j}(x)),\]
\[0_{d\times d} = \sum_x A_{kx}(\mathcal{V}_{i\rightarrow j}(k) - \sigma_{kx}\mathcal{V}_{i\rightarrow j}(x)),\,\forall k\neq i,j.\]
The equations above should be interpreted as follows: the current from external source injecting into $i$ (resp. $j$) should be equal to the total current from all edges incident to $i$ (resp. $j$); and if there is no external source, the total current from all edges incident to $k$ should sum up to $0_{d\times d}$.  See \Cref{fig: current balance} for an illustration.

In this way, given a ``unit'' voltage at $i$ (specifically: a matrix-valued voltage equal to $I_{d\times d}$) and ``zero'' voltage at $j$, there will be a current $\mathcal{C}_{ii}^\sigma$ going into $i$ and $-\mathcal{C}_{ji}^\sigma$ current going out of $j$. 
Similarly, given a ``unit'' voltage at $j$ and ``zero'' voltage at $i$, there will be a current $\mathcal{C}_{jj}^\sigma$ going into $j$ and $-\mathcal{C}_{ij}^\sigma$ current going out of $i$.

\begin{figure}[htb!]
    \centering
    \includegraphics[width = 0.5\linewidth]{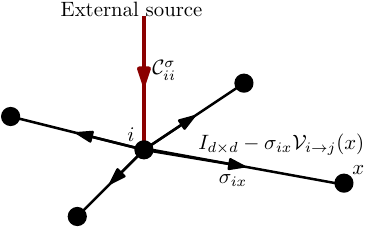}
    \caption{An illustration of the current balance equation at $i$.}
    \label{fig: current balance}
\end{figure}

\subsection{A Probabilistic Interpretation of the Connection Conductance Matrix}\label{subsec:probabilistic-interpretation-connec-conduct}

In \Cref{thm:effective-resistance-random-walk}, we recalled the relationship between effective conductance between two nodes in $G$ and the  ``escape probability" of a random walk starting at $i$. It turns out that both the connection voltage function and the connection conductance matrix can be expressed explicitly using escape probability and mean path signatures.

\begin{theorem}
    Let $(G,\sigma)$ be a connection graph and $i,j\in V$ be two fixed vertices. For any $x\in V$, we have that 
        \begin{equation}\label{eq: voltage function}
                \mathcal{V}_{i\rightarrow j}(x)=\mathbb{P}^x[T_i^0<T_j^0]\cdot \Omega^0_{xi}(j).
        \end{equation}
\end{theorem}

Alternatively, $\mathcal{V}_{i\rightarrow j}(x)$ can be characterized as follows. Let $T_{ij}^0 := T^0_{\{i, j\}}$ be a stopping time (cf. Definition \Cref{def:set-stopping-time}), and let the generalized indicator function $\chi_i:V\rightarrow\mathbb{R}^{d\times d}$ be given by $\chi(i) = I_{d\times d}$ and $\chi(x) = 0_{d\times d}$ otherwise. For each $x\in V$, we have that (see \Cref{sec:missing_proofs} for a derivation)
\begin{equation}\label{eq:alternative-voltage-function}
    \mathcal{V}_{i\rightarrow j}(x) =\mathbb{E}^x\left[ \chi_i \left(X_{T_{ij}} \right)\prod_{s=0}^{T_{ij}^0} \sigma_{X_s X_{s+1}} \right].
\end{equation}

For any $i,j\in V$, recall that we let $c_{ij}$ denote the graph effective conductance between the two nodes and let $\mathbb{P}^i[T_j^1<T_i^1]$ represent the escape probability. Then,
    \begin{theorem}\label{thm:escape probability}
        For the conductance matrix $\mathcal{C}^\sigma(i,j) = \begin{bmatrix}
                \mathcal{C}_{ii}^\sigma & \mathcal{C}_{ij}^\sigma \\
                \mathcal{C}_{ji}^\sigma & \mathcal{C}_{jj}^\sigma\end{bmatrix}$, one can interpret the blocks using escape probability and mean path signatures as follows:
        \begin{align*}
            \mathcal{C}_{ii}^\sigma & =\deg(i)\cdot\left(I_{d\times d} - (1-\mathbb{P}^i[T_j^1<T_i^1])\cdot \Omega_{ii}^1(j)\right) 
        \end{align*}
        \[\mathcal{C}_{ji}^\sigma = -\deg(j)\cdot \mathbb{P}^j[T_i^1<T_j^1]\cdot \Omega^1_{ji}(j).\]
   Then, by the fact that $c_{ij}=\deg(i)\cdot \mathbb{P}^i[T_j^1<T_i^1]=\deg(j)\cdot \mathbb{P}^j[T_i^1<T_j^1]$, we have that
        \[\mathcal{C}^\sigma(i,j) = c_{ij}\begin{bmatrix}
                I_{d\times d}  & -I_{d\times d} \\
                -I_{d\times d} & I_{d\times d}\end{bmatrix}+\begin{bmatrix}
                (\deg(i)-c_{ij})\cdot(I_{d\times d} - \Omega^1_{ii}(j)) & c_{ij}\cdot(I_{d\times d} - \Omega^1_{ij}(i))           \\
                c_{ij}\cdot(I_{d\times d} - \Omega^1_{ji}(j))           & (\deg(j)-c_{ij})\cdot(I_{d\times d} - \Omega^1_{jj}(i))\end{bmatrix}.\]
    \end{theorem}

This result successfully separates the classical effective conductance based on graph structure from the mean path signature based on graph signatures and hence helps us to appreciate the definition of the conductance matrix.
\begin{proof}
    Using the current balance equation in \Cref{subsubsec:physical-perspective-conductance}, one has that
    \begin{align*}
        \mathcal{C}_{ii}^\sigma & = \sum_{x}(\mathcal{V}_{i\rightarrow j}(i)-\sigma_{ix}\mathcal{V}_{i\rightarrow j}(x))A_{ix}                                                                              \\
        & =\deg(i)\cdot(\mathcal{V}_{i\rightarrow j}(i) - \sum_{x}\mathbb{P}_{i,x}\sigma_{ix}\mathcal{V}_{i\rightarrow j}(x))                                                       \\
        & =\deg(i)\cdot\left(I_{d\times d} - \sum_{x}\mathbb{P}_{i,x}\sigma_{ix}\mathbb{P}^x[T_i^1<T_j^1]\cdot \mathbb{E}^x[\sigma_{p_{x,x_1,x_2,\ldots,x_n,i}}|T_i^1<T_j^1]\right) \\
        & =\deg(i)\cdot\left(I_{d\times d} - \sum_{x}\mathbb{P}^i[T_i^1<T_j^1]\cdot \mathbb{E}^i[\sigma_{p_{i,x,x_1,x_2,\ldots,x_n,i}}|T_i^1<T_j^1]\right)                          \\
        & =\deg(i)\cdot\left(I_{d\times d} - \mathbb{P}^i[T_i^1<T_j^1]\cdot \Omega_{ii}^1(j)\right).
    \end{align*}
    Here in the fourth equality, the expectation is over all paths starting at $i$, going to $x$ in the next step, and finally coming back to $i$ before hitting $j$. The formula for $\mathcal{C}_{ji}^\sigma$ can be similarly derived and we defer the details to \Cref{sec:missing_proofs}.
\end{proof}

As a direct consequence, one has the following result regarding cycle graphs.

\begin{example}[Cycle Graphs]\label{ex:cycle-graphs}
    If $G$ is a cycle graph, then
    $\mathcal{C}_{ii}^\sigma=\mathcal{C}_{jj}^\sigma=c_{ij}\cdot I_{d\times d}.$
    This follows directly from the fact hat $\Omega^1_{ii}(j)=I_{d\times d}$ for any $i,j$ in a cycle graph.
    We postpone the discussion of the off diagonal entries to later \Cref{ex:cycle-graphs-continued}.
\end{example}

We also note the following useful algebraic consequence.

\begin{lemma}\label{lm:invertible Cii}
    The matrix $\mathcal{C}_{ii}^\sigma$ is invertible.
\end{lemma}

\begin{proof}
    As $\Omega_{ii}^1(j)$ is the convex combination of orthonormal matrices and is symmetric, one has that $I_{d\times d}-\Omega_{ii}^1(j)$ is positive semidefinite.
    Furthermore, we have that $c_{ij}= \deg(i)\cdot \mathbb{P}^i[T_j^1<T_i^1]\leq \deg(i)$.
    Hence, $\mathcal{C}_{ii}^\sigma = c_{ij}I_{d\times d} + (\deg(i) - c_{ij}) \cdot (I_{d\times d} - \Omega_{ii}^1(j))$ is positive definite and hence invertible.
\end{proof}

\subsection{Some Examples of Conductance Matrices}\label{subsec:conductance-examples}

In this final subsection we cover two examples of effective conductance matrices. They are counterparts to the famous series and parallel combination of resistors in classical electrical networks.

\begin{example}[Series combination]
    Consider a line graph shown in \Cref{fig:line_graph}. Then, for vertices $1$ and $n$, one has that
    \[\mathcal{C}^\sigma(1,n) = \frac{1}{\sum_{l=1}^{n-1}w_{l,l+1}^{-1}}\begin{bmatrix}
            I_{d\times d}                  & - \prod_{l=1}^{n-1}\sigma_{l,l+1} \\
            -\prod_{l=n}^{2}\sigma_{l,l-1} & I_{d\times d}\end{bmatrix}.\]
    This follows directly from the observation that $\Omega_{11}^1(n)=I_{d\times d}$ and $\Omega_{n1}^1(n)=\prod_{l=1}^{n-1}\sigma_{l,l+1}.$
\end{example}

\begin{figure}[htb!]
    \centering
    \begin{subfigure}[c]{0.45\linewidth}
        \centering
        \includegraphics[width=\linewidth]{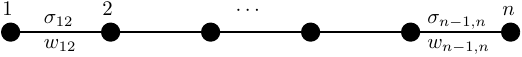}
        \caption{A line connection graph with $n$ vertices}
        \label{fig:line_graph}
    \end{subfigure}
    \hfill
    \begin{subfigure}[c]{0.45\linewidth}
        \centering
        \includegraphics[width=\linewidth]{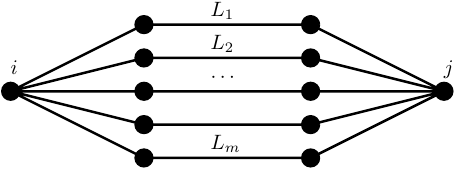}
        \caption{A parallel combination of $m$ line connection graphs $L_l$ with $l=1,\ldots,m$}
        \label{fig:parallel_graph}
    \end{subfigure}
    \caption{Series combination and parallel combination}
\end{figure}

\begin{example}[Parallel combination]
    Consider a combination of several line graphs shown in \Cref{fig:parallel_graph}.
    We numerate each line using index $l=1,\ldots,m$ and let $C^{\sigma,l}(i,j)$ denote the conductance matrix for the $l$-th line graph. Then, one has that (see \Cref{sec:missing_proofs} for a proof):
    \begin{equation}\label{eq:parallel}
        \mathcal{C}^\sigma(i,j) = \sum_{l=1}^m\mathcal{C}^{\sigma,l}(i,j).
    \end{equation}
\end{example}

Now, we use the above examples to continue our computation of the conductance matrix for a cycle graph in \Cref{ex:cycle-graphs}.

\begin{example}[Cycle Graphs - continued]\label{ex:cycle-graphs-continued}
    As shown in \Cref{ex:cycle-graphs}, if $G$ is a cycle graph, then for any distinct vertices $i,j\in V$, we have $\mathcal{C}_{ii}^\sigma=c_{ij}\cdot I_{d\times d}.$
    Now, we specify the off diagonal blocks in $\mathcal{C}^\sigma(i,j)$ using the previous examples. In this case, the two vertices result in a combination of two line graphs. Based on \Cref{prop:cycle}, we assume that only one edge incident to $i$ has the signature $\sigma$, while all other edges have the identity matrix as the signature. Consequently, we obtain:
    \[\mathcal{C}^\sigma(i,j)=\begin{bmatrix}
        (c_1+c_2) I_{d\times d} & -(c_1\sigma + c_2 I_{d\times d}) \\ -(c_1\sigma^\mathrm{T} + c_2 I_{d\times d}) & (c_1+c_2) I_{d\times d}
    \end{bmatrix}\]
    where $c_1$ is the effective conductance between $i$ and $j$ in the line graph containing one edge with signature $\sigma$ and $c_2$ is the effective conductance between $i$ and $j$ in the other line graph.
\end{example}

\section{Resistance on Connection Graphs}\label{subsec:connection-resistance-matrices}

Given the conductance matrix, one naturally wonders how to define a ``resistance matrix''. This question is more involved than its classical counterpart where the effective resistance is simply the reciprocal of the effective conductance. 
Instead of naively defining the resistance matrix as the pseudoinverse of the conductance matrix, we choose to first establish a Poisson type problem based on the Dirichlet problem studied in \Cref{eq: dirichlet boundary problem connection graph} and hence define a resistance matrix that is ``consistent'' with the classical definition in a certain sense. The resistance matrix we obtain is almost the pseudoinverse of the conductance matrix (cf. \Cref{prop:R = C}) and presents clean formulation for absolutely inconsistent signatures (cf. \Cref{prop:resistance-absolutely-inconsistent}).

Note that for the solution $\mathcal{V}_{i\to j}$ of the Dirichlet problem, we have that
\[ (\mathcal{L}\mathcal{V}_{i\rightarrow j}) (x)=\begin{cases}
        \mathcal{C}_{ii}^\sigma               & x= i             \\
        \mathcal{C}_{ji}^\sigma & x = j            \\
        0                           & \text{otherwise}\end{cases}.\]
Just as in the Poisson problem \Cref{eq: poisson problem classical} where the source terms are units, we hence normalize the right hand side by right multiplying it with $\mathcal{C}_{ii}^\sigma$ and obtain the following Poisson type problem:
\begin{equation}\label{eq:Poisson connection}
    (\mathcal{L}\mathcal{W}_{i\rightarrow j}) (x)=\begin{cases}
        I_{d\times d}               & x= i             \\
        -(\Omega_{ij}^0)^\mathrm{T} & x = j            \\
        0                           & \text{otherwise}\end{cases}.
\end{equation}

Note that the appearance of the mean path signature follows from \Cref{prop: mean path signature vs schur}.
There are two issues that immediately arise when constructing $\mathcal{W}$ in this manner: existence and uniqueness of $\mathcal{W}_{i\rightarrow j}$. The existence follows directly from the fact that $\mathcal{W}_{i\rightarrow j} := \mathcal{V}_{i\rightarrow j}\cdot (\mathcal{C}_{ii}^\sigma)^{-1}$ is a solution to \Cref{eq:Poisson connection} and hence that a solution $\mathcal{W}_{i\rightarrow j}$ exists in general.

Uniqueness does not hold in general. To resolve this we simply choose $\mathcal{W}_{i\rightarrow j}$ to be the unique solution to \Cref{eq:Poisson connection} with minimum Euclidean norm, i.e., using the pseudoinverse of $\mathcal{L}$ \cite[Ch. 3]{ben2003generalized}:
    \begin{equation}\label{eq:poisson-psuedoinv-soln}
        \mathcal{W}_{i\rightarrow j}:=\mathcal{L}^\dagger \begin{bmatrix}
            I_{d\times d}               \\
            -(\Omega_{ij}^0)^\mathrm{T} \\\hline
            0_{(n-2)d\times d}
        \end{bmatrix}.
    \end{equation}
Henceforth for any fixed $i,j$ we use the notation $W_{i\rightarrow j}$ to refer to the specific solution constructed in the manner \Cref{eq:poisson-psuedoinv-soln}.

Furthermore, by direct calculation, similar to the proof of \Cref{prop: mean path signature vs schur}, we find that
\begin{equation}\label{eq:half of resistance matrix}
    \begin{bmatrix}
        \mathcal{W}_{i\rightarrow j}(i)  \\
        \mathcal{W}_{i\rightarrow j} (j)
    \end{bmatrix} = (\mathcal{L}/\mathcal{L}_{\{i,j\}^c})^\dagger \begin{bmatrix}
        I_{d\times d}               \\
        -(\Omega_{ij}^0)^\mathrm{T}
    \end{bmatrix}.
\end{equation}

With the setup of $\mathcal{W}_{i\rightarrow j}$ in hand, we define the resistance matrix as follows:

\begin{definition}\label{def: resistance matrix}
    For any $i,j\in V$, we define the resistance matrix $\mathcal{R}^\sigma(i,j)$ as follows:
    \[\mathcal{R}^\sigma(i,j):=\begin{bmatrix}
            \mathcal{W}_{i\rightarrow j}(i)  & \mathcal{W}_{j\rightarrow i}(i) \\
            \mathcal{W}_{i\rightarrow j} (j) & \mathcal{W}_{j\rightarrow i}(j)
        \end{bmatrix}.\]
\end{definition}

In the underlying graph $G$, as described in \Cref{subsec:effective-resistance}, the effective conductance and resistance are related by reciprocal: $c_{ij} = r_{ij}^{-1}$. The following result shows that the resistance matrix is almost the pseudoinverse of the conductance matrix and it follows directly from \Cref{eq:half of resistance matrix} and the definition of $\mathcal{C}^\sigma(i,j)$.

\begin{proposition}\label{prop:R = C}
    For any $i,j\in V$, one has that
    \[\mathcal{R}^\sigma(i,j) = \mathcal{C}^\sigma(i,j)^\dagger \begin{bmatrix}
            I_{d\times d}               & -(\Omega_{ji}^0)^\mathrm{T} \\
            -(\Omega_{ij}^0)^\mathrm{T} & I_{d\times d}
        \end{bmatrix}.\]
\end{proposition}

To justify \Cref{def: resistance matrix}, we present a physical interpretation of the resistance matrix. In \Cref{subsubsec:physical-perspective-conductance}, we offer an interpretation of the conductance matrix as a representation of currents flowing between vertices $i$ and $j$ when ``unit" voltages are applied. Building upon this, we provide a dual interpretation of the resistance matrix. Specifically, when a source at vertex $i$ generates a ``unit" current (i.e., a matrix-valued current with a value of $I_{d\times d}$) and a sink at vertex $j$ receives a current of $\Omega_{ij}^\mathrm{T}$, the resistance matrix records the corresponding voltages at vertices $i$ and $j$.

\subsection{Resistance Matrix under Equivalence and Direct Sum}\label{subsec:resistance-direct-sum}

In this section we discuss how the effective resistance matrices operate at the level of signature equivalence classes and direct sums of signatures. Most proofs are elementary and can be found in \Cref{sec:missing_proofs}.

\begin{proposition}[Resistance matrix under equivalence]\label{prop:resistance-equivalence}
    Assume that $\sigma\simeq\tau$ and  let $f:V\to\mathsf{O}(d)$ be a switching map.
    Then for any $i,j\in V$, one has that
    \[ F_{ij}\mathcal{R}^\sigma(i,j)=\mathcal{R}^\tau(i,j)F_{ij},\,\text{ where } \,F_{ij}:=\begin{bmatrix}
            f(i)          & 0_{d\times d} \\
            0_{d\times d} & f(j)
        \end{bmatrix}.\]
\end{proposition}

\begin{proposition}[Resistance matrix under direct sum]\label{prop:resistance-direct-sum}
    Let $\sigma$ and $\tau$ be signatures on $G$. Then for any $i,j\in V$, one has that $\mathcal{R}^{\sigma\oplus\tau}(i,j)$ is similar to $\mathcal{R}^\sigma(i,j)\oplus\mathcal{R}^\tau(i,j)$.
\end{proposition}

\begin{example}[Consistent Graphs]\label{ex:resistance-matrix-consistent-case}
    When $(G,\sigma)$ is consistent, by \Cref{ex:consistent signature} and \Cref{prop:conductance matrix under direct sum}, we have $\sigma\simeq\bigoplus_{i=1}^d\iota^1$. Then, for any $i,j\in V$, one has that
    \[\mathcal{R}^\sigma(i,j)\simeq \bigoplus_{i=1}^d\frac{1}{2}\begin{bmatrix}
            r_{ij}  & -r_{ij} \\
            -r_{ij} & r_{ij}
        \end{bmatrix}.\]
    This follows from
    \[\mathcal{C}^\sigma(i,j)^\dagger \begin{bmatrix}
        I_{d\times d}               & -(\Omega_{ji}^0)^\mathrm{T} \\
        -(\Omega_{ij}^0)^\mathrm{T} & I_{d\times d}
    \end{bmatrix} \simeq \bigoplus_{i=1}^d\begin{bmatrix}
        c_{ij}  & -c_{ij} \\
        -c_{ij} & c_{ij}
    \end{bmatrix}^\dagger\begin{bmatrix}
        1  & -1 \\
        -1 & 1
    \end{bmatrix}=\bigoplus_{i=1}^d\frac{1}{2}\begin{bmatrix}
        r_{ij}  & -r_{ij} \\
        -r_{ij} & r_{ij}
    \end{bmatrix}\]
\end{example}

Similarly to the case of conductance matrices, given the decomposition (cf. \Cref{thm:decomposition}) of any signature $\sigma$ on $G$: $\sigma\simeq (\bigoplus_{i=1}^\rho\iota^1)\oplus \tau$,
where $\tau$ is absolutely inconsistent, using \Cref{ex:resistance-matrix-consistent-case}, one obtains the following characterization of $\mathcal{R}^\sigma(i,j)$:
    \[\mathcal{R}^\sigma(i,j)\simeq \left(\bigoplus_{i=1}^\rho\frac{1}{2}\begin{bmatrix}
            r_{ij}  & -r_{ij} \\
            -r_{ij} & r_{ij}
        \end{bmatrix}\right)\oplus \mathcal{R}^\tau(i,j).\]

We now study $\mathcal{R}^\sigma(i,j)$ when the signature is absolutely inconsistent.
\begin{proposition}\label{prop:resistance-absolutely-inconsistent}
    Let $\sigma$ be absolutely inconsistent. Then, for any $i,j\in V$, 
    \[\mathcal{R}^\sigma(i,j)=\begin{bmatrix}
            (C_{ii}^\sigma)^{-1} & 0_{d\times d}        \\
            0_{d\times d}        & (C_{jj}^\sigma)^{-1}
        \end{bmatrix}.\]
\end{proposition}
\begin{proof}
    As we know that $\mathcal{V}_{i\rightarrow j}\cdot (\mathcal{C}_{ii}^\sigma)^{-1}$ is a solution to \Cref{eq:Poisson connection} such that $(\mathcal{V}_{i\rightarrow j}\cdot (\mathcal{C}_{ii}^\sigma)^{-1})(j)=0_{d\times d}$.
    When $\sigma$ is absolutely inconsistent, we know that $\mathcal{L}^\sigma$ is invertible. Hence, $\mathcal{V}_{i\rightarrow j}\cdot (\mathcal{C}_{ii}^\sigma)^{-1}$ is the unique solution to \Cref{eq:Poisson connection}. Note that $\mathcal{V}_{i\rightarrow j}(i)=I_{d\times d}$ and $\mathcal{V}_{i\rightarrow j}(j)=0_{d\times d}$. Therefore, we have that
    \[\mathcal{W}_{i\to j}(i)=\mathcal{V}_{i\rightarrow j}(i)\cdot (\mathcal{C}_{ii}^\sigma)^{-1}=(\mathcal{C}_{ii}^\sigma)^{-1}\quad\text{and}\]
    \[\mathcal{W}_{i\to j}(j)=\mathcal{V}_{i\rightarrow j}(j)\cdot (\mathcal{C}_{ii}^\sigma)^{-1}=0_{d\times d}.\]
\end{proof}

\subsection{Connection Resistance: A Scalar Version of the Resistance Matrix}\label{sec:scalar resistance matrix}
Finally, we would like to end this paper with a definition of a scalar version of the resistance matrix.
Recall that in \Cref{subsubsec:energy-perspective} we showed how the effective resistance appears as the Dirichlet energy of a solution to the Poisson problem $Lf=e_i-e_j$ (cf. \Cref{eq: poisson problem classical}):
$r_{ij}=E(f)=f^\mathrm{T}Lf.$

Motivated by this insight, one might define the {scalar connection resistance} between $i$ and $j$ as the Dirichlet energy of a solution to \Cref{eq:Poisson connection}:
$\frac{1}{2}\tr (\mathcal{W}_{i\to j}^\mathrm{T}\mathcal{L}\mathcal{W}_{i\to j})$. We note, however, that this term is asymmetric in $i$ and $j$. To ensure symmetry, it is natural to consider $\frac{1}{2}\tr (\mathcal{W}_{i\to j}^\mathrm{T}\mathcal{L}\mathcal{W}_{i\to j}+\mathcal{W}_{j\to i}^\mathrm{T}\mathcal{L}\mathcal{W}_{j\to i})$, which turns out to be our final definition up to certain normalization\footnote{A normalization is necessary as if the $d$-dim signature is consistent, this term coincides with $d\cdot r_{ij}$.}. 

\begin{definition}[Connection effective resistance]\label{def:scalar_connection_er}
    For any $i,j\in V$, we define the \emph{connection effective resistance} between $i$ and $j$ as
    \[r_{ij}^\sigma := \frac{1}{2d}\tr\left(\mathcal{W}_{i\to j}^\mathrm{T}\mathcal{L}\mathcal{W}_{i\to j}+\mathcal{W}_{j\to i}^\mathrm{T}\mathcal{L}\mathcal{W}_{j\to i}\right).\]
\end{definition}

Just as the classical effective resistance can be written as $r_{ij}=(e_i-e_j)^\mathrm{T}L^\dagger(e_i-e_j)$, we can also characterize the connection effective resistance is a similar manner.
\begin{proposition}\label{def:trace-resistance}
    If we let    
$N_{ij}=\begin{bmatrix}
    0_{d\times d}, \cdots,
    I_{d\times d},\cdots, -\Omega_{ij}^0,\cdots
    ,0_{d\times d}\end{bmatrix}^\mathrm{T}$ for any $i,j\in V$, 
   then we have that
    \[r_{ij}^\sigma = \frac{1}{2d}(\tr(N_{ij}^\mathrm{T}\mathcal{L}^\dagger N_{ij})+\tr(N_{ji}^\mathrm{T}\mathcal{L}^\dagger N_{ji})).\]
\end{proposition}

We next specify some relationship between $N_{ij}^\mathrm{T}\mathcal{L}^\dagger N_{ij}$ and the conductance matrix and an explicit formula for computing the connection effective resistance.

\begin{lemma}\label{prop:resistance}
    Let $\sigma\simeq (\bigoplus_{i=1}^\rho\iota^1)\oplus \tau$ where $\tau$ is absolutely inconsistent. Then, 
    \[N_{ij}^\mathrm{T}\mathcal{L}^\dagger N_{ij}=(\mathcal{C}^\sigma_{ii})^{-1}\simeq  \left(\bigoplus_{i=1}^\rho[r_{ij}]\right)\oplus (\mathcal{C}^\tau_{ii})^{-1} \text{ and }\]
    \[N_{ji}^\mathrm{T}\mathcal{L}^\dagger N_{ji}=(\mathcal{C}^\sigma_{jj})^{-1}\simeq  \left(\bigoplus_{i=1}^\rho[r_{ij}]\right)\oplus (\mathcal{C}^\tau_{jj})^{-1}.\]
\end{lemma}
\begin{proof}
    We prove the first equality below (the same proof applies to $N_{ji}^\mathrm{T}\mathcal{L}^\dagger N_{ji}$) and the rest follows from the decomposition results of blocks in conductance matrices (see \Cref{prop:conductance matrix under direct sum} and the discussion thereafter).

    Recall from \Cref{prop: mean path signature vs schur} that $-\Omega_{ij}^0=(\mathcal{C}_{ii}^\sigma)^{-1}\mathcal{C}_{ij}^\sigma$. Without loss of generality, we assume that $i<j$. Then, we have that
    \begin{align*}
        N_{ij}^\mathrm{T}\mathcal{L}^\dagger N_{ij}& = [I_{d\times d},-\Omega_{ij}^0]\mathcal{C}^\sigma(i,j)^\dagger[I_{d\times d},-\Omega_{ij}^0]^\mathrm{T}\\
        &=[(\mathcal{C}_{ii}^\sigma)^{-1}, 0_{d\times d}]\mathcal{C}^\sigma(i,j)^\mathrm{T}\mathcal{C}^\sigma(i,j)^\dagger \mathcal{C}^\sigma(i,j)[(\mathcal{C}_{ii}^\sigma)^{-1}, 0_{d\times d}]^\mathrm{T}\\
        &=[(\mathcal{C}_{ii}^\sigma)^{-1}, 0_{d\times d}]\mathcal{C}^\sigma(i,j)\mathcal{C}^\sigma(i,j)^\dagger \mathcal{C}^\sigma(i,j)[(\mathcal{C}_{ii}^\sigma)^{-1}, 0_{d\times d}]^\mathrm{T}\\
        &=[(\mathcal{C}_{ii}^\sigma)^{-1}, 0_{d\times d}]\mathcal{C}^\sigma(i,j)[(\mathcal{C}_{ii}^\sigma)^{-1}, 0_{d\times d}]^\mathrm{T}=(\mathcal{C}_{ii}^\sigma)^{-1}.
    \end{align*}
\end{proof}
Based on this result, we establish the following explicit formula for the effective resistance under decomposition of signatures.

\begin{theorem}\label{thm:resistance decomposition}
     Given $(G,\sigma)$, let $1\leq \rho\leq d$ be the dimension of the kernel of $\mathcal{L}^\sigma$, and let $\tau$ be the absolutely inconsistent component of $\sigma$. Then it holds:
        \[r^\sigma_{ij} = \frac{\rho}{d}r_{ij} + \frac{1}{2d}\tr\left(\left(\mathcal{C}^\tau_{ii}\right)^{-1}+(\mathcal{C}^\tau_{jj})^{-1}\right).\]
\end{theorem}

\begin{proof}
    This follows from applying \Cref{prop:resistance} to \Cref{def:trace-resistance}.
\end{proof}

One immediate consequence of the above theorem is that the effective resistance is invariant under equivalence of signatures: this follows from \Cref{prop:conductance matrix under equivalence} and the fact that the trace of a matrix is invariant under similarity transformations. Another direct result is that when the signature is consistent, then the connection effective resistance is equal to the classical definition of the effective resistance.

Finally, we establish that the effective resistance is continuous with respect to change of signatures.

\begin{theorem}\label{thm:rijsigma_is_continous}
    Given $(G,\sigma)$, let $i,j\in V$ be fixed nodes, and let $r_{ij}$ be the effective resistance between $i,j$ for the underlying graph $G$. Then the function $\sigma\mapsto r^\sigma_{ij}$ is continuous.
\end{theorem}
\begin{proof}
    Note  from \Cref{prop:resistance}
     that 
    \begin{equation}\label{eq:rijsigma_formula}
        r_{ij}^\sigma=\frac{1}{2d}\left(\tr((\mathcal{C}_{ii}^\sigma)^{-1})+\tr((\mathcal{C}_{jj}^\sigma)^{-1})\right).
    \end{equation}
    Then, the result follows from continuity and invertibility of both $\mathcal{C}_{ii}^\sigma$ and $\mathcal{C}_{jj}^\sigma$.
\end{proof}

Note that the new connection resistance is defined for every pair of vertices and is continuous with respect to the signature. This is in contrast to the definition by \cite{chung2014connection} and help justify our definition.

We end with one counterintuitive property of $r_{ij}^\sigma$ comparing the connection resistance with the effective resistance of underlying graphs. The proof is in \Cref{sec:missing_proofs}.

\begin{proposition}\label{prop:resistance-lower-bound}
    For any $i,j\in V$, it holds that $r_{ij}^\sigma\leq r_{ij}$. 
\end{proposition}
The equality does not hold in general (see experiments in Section \ref{section:numerical_exp_scalar_cr}).
This result indicates that the presence of inconsistency in the signature reduces the energy of the solution to the Poisson problem, which is counterintuitive. Inconsistency is typically seen as obstacles in random walks within connection graphs, which should typically increase the energy. The interpretation of this phenomenon remains an open problem.

\subsubsection{Numerical Experiments on the Connection Resistance}\label{section:numerical_exp_scalar_cr}
We conduct three numerical experiments contrasting our connection resistance $r_{ij}^\sigma$, as defined in Definition \ref{def:scalar_connection_er} and implemented using \Cref{eq:rijsigma_formula}, with the standard effective resistance and the connection resistance measure proposed by Chung et al.~\cite{chung2014connection}. These experiments involve a dumbbell graph (\Cref{fig:dumbbell_graph}) and a Wheatstone-bridge graph (\Cref{fig:wheatstonebridge_graph}), where we assign 3-dimensional signatures to specific edges.

For these experiments, most edges are assigned the identity matrix $I_{3\times 3}$, with chosen edges assigned signatures of the form:
\begin{equation*}
    \sigma_{ij}(\theta) = 
    \begin{pmatrix}
    1 & 0 & 0 \\
    0 & \cos(\theta) & -\sin(\theta) \\
    0 & \sin(\theta) & \cos(\theta)
    \end{pmatrix}.
\end{equation*}

In the case of the dumbbell graph (\Cref{fig:dumbbell_graph}), the edge $(1,2)$ is assigned a signature $\sigma_{12}$, with $\theta_{12}$ sampled from a grid over the interval $[0, 2\pi]$. Beyond edge $(1,2)$, edge $(2,3)$ is given two different signature configurations: $\theta_{23}$ is chosen to be either $0$ or $\pi/2$.
For the Wheatstone-bridge graph (\Cref{fig:wheatstonebridge_graph}), the edge $(2,4)$ is assigned a signature $\sigma_{24}$, with $\theta_{24}$ sampled from a grid on the interval $[0, 2\pi]$.

\begin{figure}[htb!]
    \centering
    \begin{subfigure}[c]{0.7\linewidth}
        \centering
        \includegraphics[width=\linewidth]{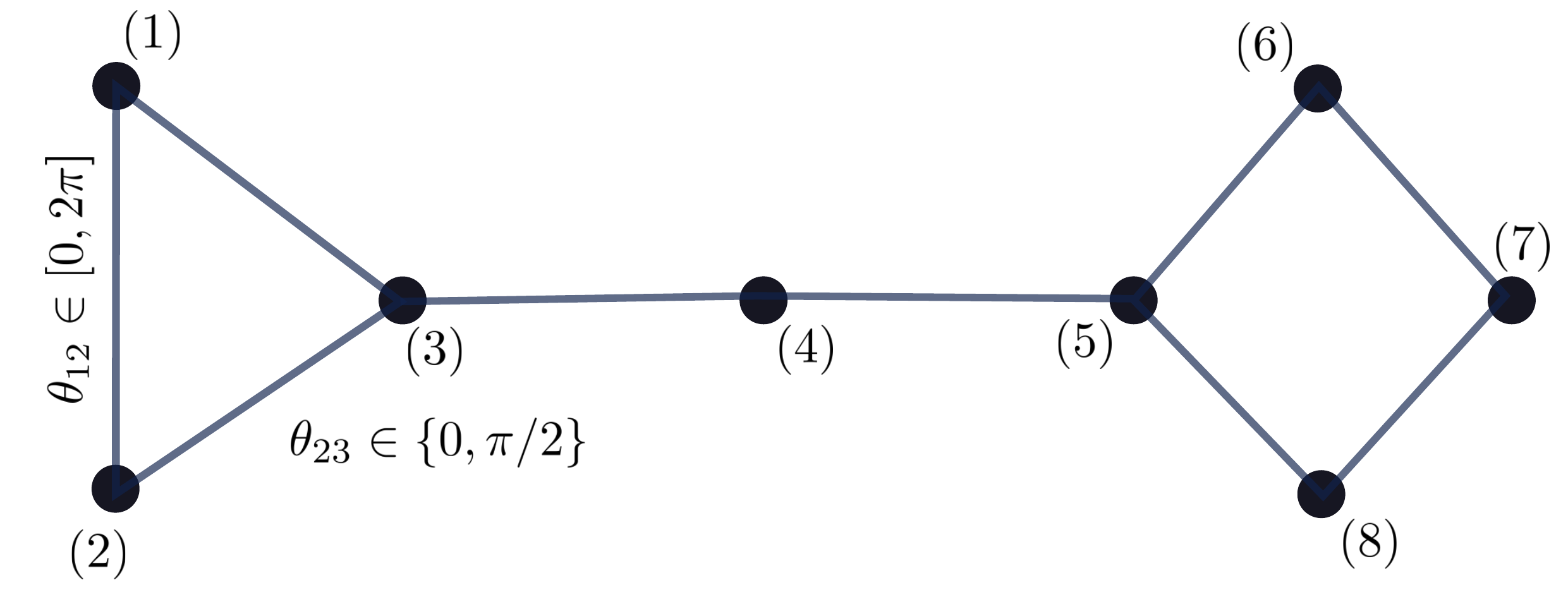}
        \caption{Dumbbell graph}
        \label{fig:dumbbell_graph}
    \end{subfigure}
    \hfill
    \begin{subfigure}[c]{0.25\linewidth}
        \centering
        \includegraphics[width=\linewidth]{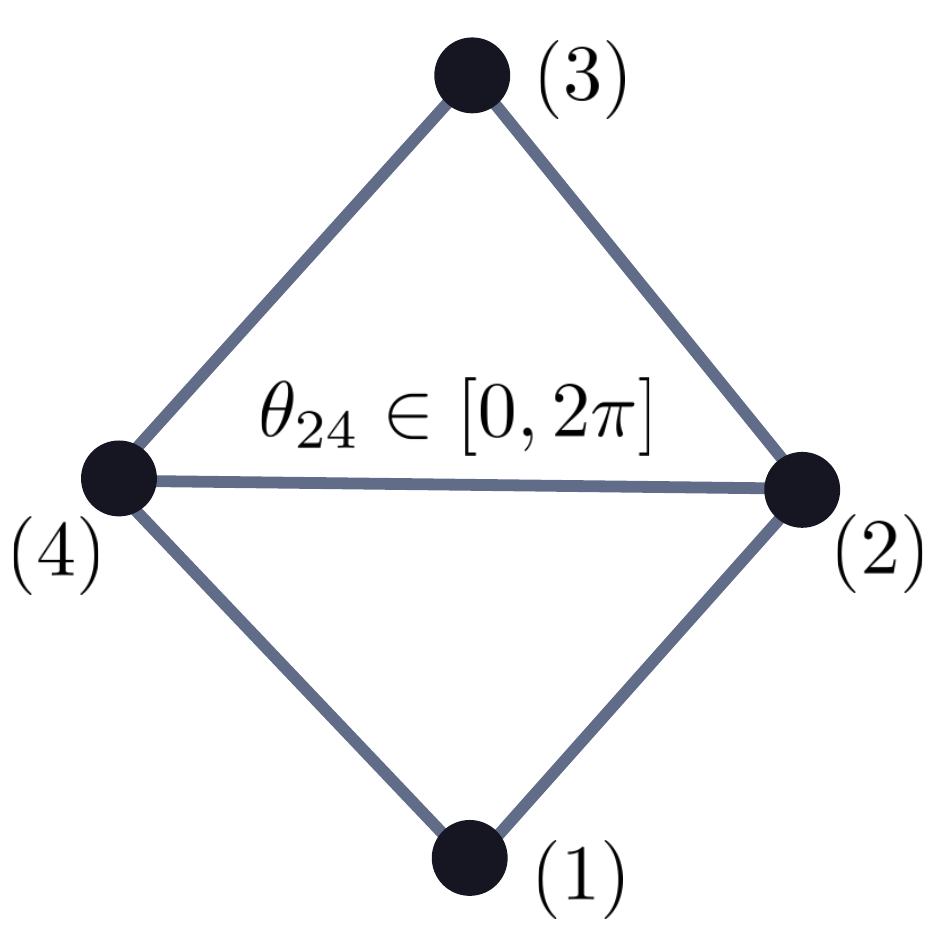}
        \caption{Wheatstone-bridge graph}
        \label{fig:wheatstonebridge_graph}
    \end{subfigure}
    \caption{Illustration of graphs considered in the numerical experiments}
    \label{fig:dumbbell_and_wheatstone_graphs_illustration}
\end{figure}

We observe the following properties of $r_{ij}^\sigma$ from these experiments. 
\begin{enumerate}
  \item $r_{ij}^\sigma$ varies continuously with signature changes, unlike the CR from \cite{chung2014connection}.
  \item $r_{ij}^\sigma\leq r_{ij}$ in accordance with Proposition \ref{prop:resistance-lower-bound};
  \item $r_{ij}^\sigma = r_{ij}$ when the graph is consistent.
\end{enumerate}

\begin{figure}[htb!]
    \centering
    \begin{subfigure}[c]{\linewidth}
        \centering
        \includegraphics[width=\linewidth]{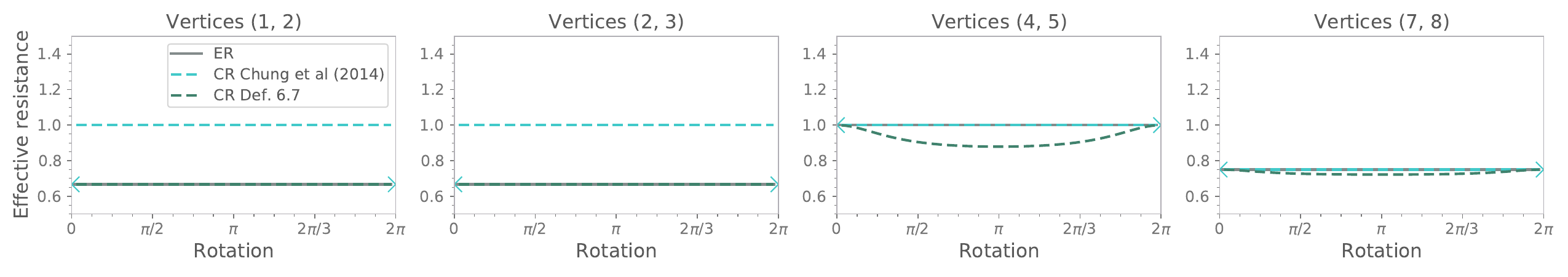}
        \caption{Dumbbell graph}
        \label{fig:exp_dumbbell_graph}
    \end{subfigure}
    \hfill
    \begin{subfigure}[c]{\linewidth}
        \centering
        \includegraphics[width=\linewidth]{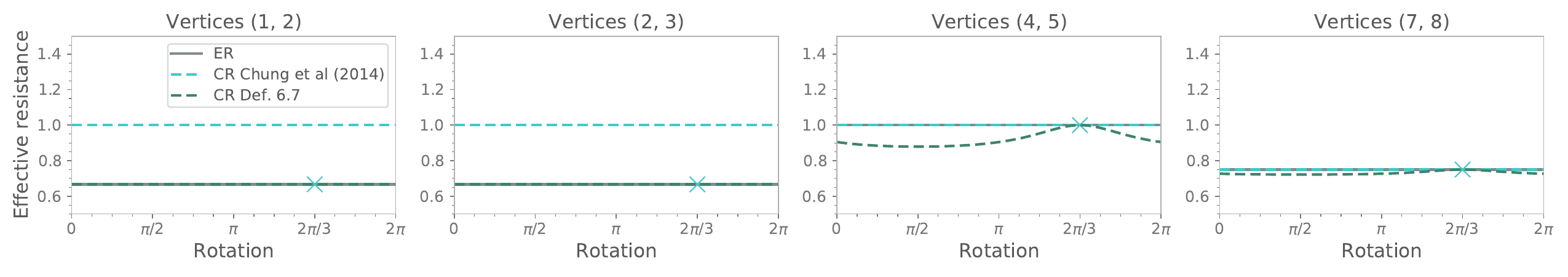}
        \caption{Dumbbell graph with fixed rotation $\theta= 90^o$ for the signature on edge $(2,3)$.}
        \label{fig:exp_dumbbell_graph_90deg}
    \end{subfigure}
    \hfill
    \begin{subfigure}[c]{\linewidth}
        \centering
        \includegraphics[width=\linewidth]{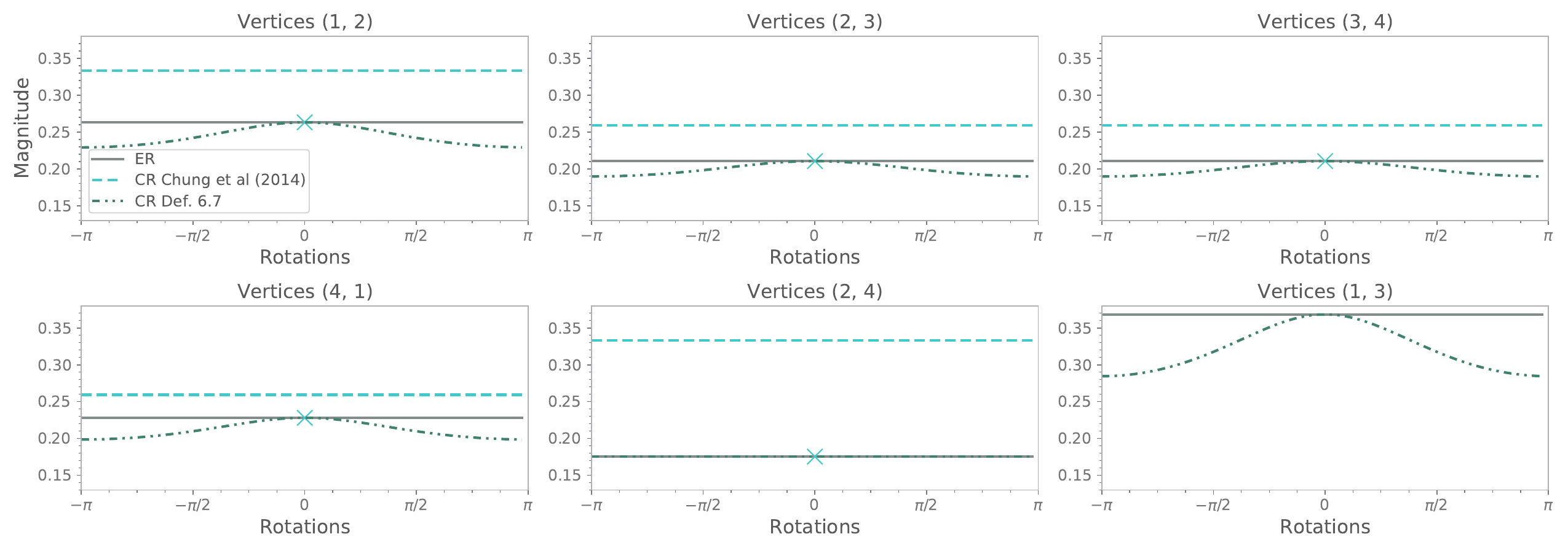}
        \caption{Wheatstone-bridge graph}
        \label{fig:exp_wheatstonebridge_graph}
    \end{subfigure}
    \caption{\textbf{Numerical comparison of different notions of effective resistance.}  For clarity, `ER' in the figures corresponds to the standard effective resistance, while `CR' refers to the connection effective resistance as per Definition \ref{def:scalar_connection_er} and Chung et al.~\cite{chung2014connection}, respectively.
    (a) A dumbbell graph where $\theta$ on edge $(1,2)$ is varied over $[0, 2\pi]$. 
    (b) A dumbbell graph similar to (a), but with $\theta$ fixed at $\pi/2$ for edge $(2,3)$. 
    (c) A Wheatstone-bridge graph where $\theta$ for the edge signature on edge $(1,3)$ is varied over $[0, 2\pi]$.}
    \label{fig:numerical_experiments}
\end{figure}

\section{Discussion}
We introduced a novel concept of effective resistance specifically tailored for connection graphs, featuring desirable attributes like continuity relative to graph signature and invariance under signature equivalence. Several potential research avenues below conclude our paper. 

\paragraph{Properties of the Connection Resistance} 
While we have established certain properties, further examination of connection resistance is intriguing. For instance, given that graph effective resistance is a metric, we can explore if this extends to connection resistance. A probabilistic viewpoint on connection resistance, perhaps via a commute time notion for connection graphs, could also be valuable.

\paragraph{Graph cut and Cheeger inequality in connection graphs} 
Prior work has attempted to define Cheeger constants and establish related inequalities for connection graphs. These methods separate graph structures from signatures. Following \cite{memoli2022persistent}, which links graph Cheeger constant to effective conductance, we are interested in seeing if connection conductance/resistance can be similarly used to define Cheeger constants.

\paragraph{Analysis of Graph Neural Networks} We note that the increasing adoption of generalized graphs, including magnetic graphs and connection graphs, in the development of neural networks for handling complex data has been a recent trend \cite{zhang2021magnet,barbero2022sheaf}. Inspired by these advancements, it is intriguing to consider the potential application of our notion of effective resistance in analyzing and understanding such neural networks. By leveraging insights from recent work in \cite{arnaiz2022diffwire,digiovanni2023oversquashing,black2023understanding}, we can explore the impact of effective resistance on network behavior and performance optimization.

\section*{Acknowledgements}

This work was supported by funding from the NSF (CCF 2217033 to GM, ZW and YW, CCF 2112665 to YW, DMS 2012266 to AC), the NIH (1RF1MH125317 to ZW and YW), a gift from Intel Research (to AC), Simula Research Laboratory (to AO), and the HDSI Graduate Prize Fellowship (to SR).

%%%%%%%%%%%%%%%%%%%%%%%%%%%%%%%%%%%%%%%%%%%%%%%%%%%%%%%%%%%%
\bibliographystyle{plain}
\bibliography{references}

%%%%%%%%%%%%%%%%%%%%%%%%%%%%%%%%%%%%%%%%%%%%%%%%%%%%%%%%%%%%

\appendix

%%%%%%%%%%%%%%%%%%%%%%%%%%%%%%%%%%%%%%%%%%%%%%%%%%%%%%%%%%%%%%%%%%%%%%%%%%%%%%%%

\section{Missing Proofs}\label{sec:missing_proofs}

\begin{proof}[Proof of \Cref{prop:energy-min}]
    Since (DP) admits a unique solution per \Cref{prop:maximum-norm-principle}, we know that $\mathcal{L}$ is strictly positive definite on the region
    \begin{align}
        \mathcal{F} = \{f\in\ell^2(V;\mathbb{R}^{d\times d}):f|_{\partial H} = \phi\}
    \end{align}
    Therefore, $E(f)$ is strictly convex on $\mathcal{F}$ and hence has a unique minimizer $f_0$.

    Notice that if $\mathcal{L} f_0|_{H} = 0$, then the gradient of $\nabla_f E(f_0) = 0_{nd\times d}$ and by the strict convexity of $E(f)$ on $\mathcal{F}$, $f_0$ minimizes $E(f)$.

    On the other hand, if $f_0$ solves (EP), it remains to show that $\mathcal{L}f_0|_H = 0_{d\times d}$. Let $w\in\ell^2(V;\mathbb{R}^{d\times d})$ be any function which satisfies $w|_{\partial H} = 0_{d\times d}$. Then $(f_0 + w)|_{\partial H} = \phi$, and since $f_0$ minimizes $E(f)$, we have
    \[\frac{d}{dh}E(f_0+hw)\Big{|}_{h=0} = 0.\]
    We calculate, for fixed $h\in\R$,
    \begin{align} 
        E(f_0+hw) & = \frac{1}{2}\tr\left( (f_0+hw)^\mathrm{T} \mathcal{L} (f_0+hw) \right)                                                                               \\
                  & =\frac{1}{2}\tr\left( f_0^\mathrm{T}\mathcal{L}f_0 + hw^\mathrm{T}\mathcal{L} f_0 + hf_0^\mathrm{T}\mathcal{L}w + h^2 w^\mathrm{T}\mathcal{L}w\right)
    \end{align}
    The derivative $\frac{d}{dh}(\cdot)|_{h=0}$ will only recover the two terms from $E(f_0+hw)$ linear in $h$. Therefore by symmetry of $\mathcal{L}$ and the cyclic invariance property of trace,
    \begin{align}
        \frac{d}{dh}E(f_0+hw)\Big{|}_{h=0} & = \tr\left(w^T \mathcal{L} f_0\right) = 0.
    \end{align}
    Equivalently stated in the Hilbert-Schmidt inner product on $\R^{nd\times d}$,  $\langle w, \mathcal{L} f_0\rangle = 0.$ Since $w|_{\partial H} = 0_{d\times d}$, the inner product only depends on the values of each function on the interior of $H$. Moreover, our choice of $w$ did not specify its values on the interior of $H$, so it follows that $\mathcal{L} f_0 = 0_{d\times d}$ on $H$.
\end{proof}

\begin{proof}[Proof of \Cref{prop:mean-path-product-transpose}]
    For any path $X_0=i,X_1,\ldots,X_k=j$ such that $X_l\neq i,j$ for all $l\in\{1,\ldots,k\}$, we have that its inverse path $Y_l:=X_{k-l}$ satisfies that
    \begin{enumerate}
        \item $Y_0=j$ and $Y_k=i$;
        \item $Y_l\neq i,j$ for all $l\in\{1,\ldots,k\}$.
    \end{enumerate}
    The path $X_0=i,X_1,\ldots,X_k=j$ contributes to $\Omega_{ij}^s(i)$ the term $\sigma_{iX_1}\cdots\sigma_{X_{k-1}j}$, whereas the inverse path $Y_0=j,Y_1,\ldots,Y_k=i$ contributes to $\Omega_{ji}^s(j)$ the term $\sigma_{jY_1}\cdots\sigma_{Y_{k-1}i}=(\sigma_{iX_1}\cdots\sigma_{X_{k-1}j})^\mathrm{T}$. In this way, it is direct to see that
    \[\Omega_{ij}^s(i)  = (\Omega_{ji}^s(j))^\mathrm{T}.\] 
\end{proof}

%%%%%%%%%%%%%%%%%%%%%%%%%%%%%%%%%%%%%%%%%%%%%%%%%%%%%%%%%%%%%%%%%%%%%%%%%%%%%%%%

\begin{proof}[Proof of \Cref{thm:absolutely inconsistent via mean path signature}]
    Note that for any $i\in V$, since $\Omega_i^1$ is symmetric and positive semidefinite. Hence, it suffices to prove that $\sigma$ is absolutely inconsistent iff for all $i\in V$ one has that eigenvalues of $\Omega_i^1$ are strictly smaller than $1$.

    Recall that
    \[\Omega_i^1 = \mathbb{E}\left[\prod_{\ell=1}^{T_{i}^1}\sigma_{X_{\ell-1}X_{\ell}}\middle|X_0=i\right].\]
    Hence,
    \[\|\Omega_i^1\|_2 = \left\|\mathbb{E}\left[\prod_{\ell=1}^{T_{i}^1}\sigma_{X_{\ell-1}X_{\ell}}\middle|X_0=i\right]\right\|_2\leq  \mathbb{E}\left[\left\|\prod_{\ell=1}^{T_{i}^1}\sigma_{X_{\ell-1}X_{\ell}}\right\|_2\middle|X_0=i\right]= 1.\]
    Here the first inequality follows from Jensen's inequality and the second inequality follows from the fact that $\|\sigma_{ij}\|_2=1$ for all $(i,j)\in E^\text{or}$.
    The equality above holds if and only if $\prod_{\ell=1}^{T_{i}^1}\sigma_{X_{\ell-1}X_{\ell}}$ is a constant almost surely given $X_0=i$. This implies that 
    \[\prod_{\ell=1}^{N}\sigma_{x_{\ell-1}x_{\ell}}=I_{d\times d}\]
    for any path $(i=x_0,\ldots,x_N=i)$ such that $N\geq 1$ and $x_\ell\neq i$ for all $\ell\neq 0,N$. Hence, for any $j\neq i$, we have that for any path $(i=x_0,\ldots,x_N=j)$, one has that $\prod_{\ell=1}^{N}\sigma_{x_{\ell-1}x_{\ell}}$ is a constant, i.e., is independent of the choice of path. We denote this constant by $\sigma_{ij}$ for all $j\neq i$. In this way, we construct $f:V\to\R^d$ by letting $f(i):=e_i$ and $f(j):=\sigma_{ij}e_i$ for all $j\neq i$. It is direct to check that $\mathcal{L}f=0$ and hence $\sigma$ is not absolutely inconsistent. This concludes the proof.
\end{proof}

\begin{proof}[Proof of \Cref{prop:conductance matrix under equivalence}]
    Recall by \Cref{prop:equivalence connection Laplacian} that $F\mathcal{L}^\sigma F^{-1}=\mathcal{L}^\tau$ where $F$ is the block diagonal matrix whose $i$th block is $f(i)$.
    Then, it follows from \cite[Proposition 9]{gulen2023generalization} that
    \[F_{ij} \mathcal{L}^\sigma/\mathcal{L}^\sigma_{\{ij\}^c} F_{ij}^{-1}=(F_{ij} \mathcal{L}^\sigma F_{ij}^{-1})/(F_{\{ij\}^c}\mathcal{L}^\sigma_{\{ij\}^c}F_{\{ij\}^c}^{-1})=\mathcal{L}^\tau/\mathcal{L}^\tau_{\{ij\}^c},\]
    where $F_{\{ij\}^c}$ is the block diagonal matrix whose $l$th block is $f(l)$ for all $l\neq i,j$. This implies that
    \[F_{ij} \mathcal{C}^\sigma(i,j)=\mathcal{C}^\tau(i,j) F_{ij}.\]
\end{proof}

\begin{proof}[Proof of \Cref{prop:conductance matrix under direct sum}]
    This follows from \Cref{rmk:block-decomposition-direct-sum} and the proof of \Cref{prop:conductance matrix under equivalence}. 
\end{proof}

\begin{proof}[Proof of \Cref{eq:alternative-voltage-function}]
    This can be seen easily as follows.
    \begin{align*}
        \mathbb{E}\left[ \chi_i \left(X_{T_{ij}^0} \right)\prod_{s=0}^{T_{ij}} \sigma_{X_s X_{s+1}} | X_0 = x  \right] & = \mathbb{E}\left[ \chi_i \left( i\right)\prod_{s=0}^{T_{ij}^0} \sigma_{X_s X_{s+1}} | X_0 = x, X_{T_{ij}^0} = i  \right]\mathbb{P}^x\left[X_{T_{ij}^0} = i\right] \\                        & +\mathbb{E}\left[ \chi_i \left( j \right)\prod_{s=0}^{T_{ij}^0} \sigma_{X_s X_{s+1}} | X_0 = x, X_{T_{ij}^0} = j  \right]\mathbb{P}^x\left[X_{T_{ij}^0} = j\right] \\
        & =\mathbb{E}\left[ \prod_{s=0}^{T_{ij}^0} \sigma_{X_s X_{s+1}} | X_0 = x, X_{T_{ij}^0} = i  \right]\mathbb{P}^x\left[X_{T_{ij}^0} = i\right].
    \end{align*}
\end{proof} 

\begin{proof}[The Missing Part of Proof of \Cref{thm:escape probability}]
    \begin{align*}
        \mathcal{C}_{ji}^\sigma & = \sum_{x}(\mathcal{V}_{i\rightarrow j}(j)-\sigma_{jx}\mathcal{V}_{i\rightarrow j}(x))A_{jx}                                                   \\
        & =\deg(j)\cdot(\mathcal{V}_{i\rightarrow j}(j) - \sum_{x}\mathbb{P}_{j,x}\sigma_{jx}\mathcal{V}_{i\rightarrow j}(x))                            \\
        & =-\deg(j)\cdot \sum_{x}\mathbb{P}_{j,x}\sigma_{jx}\mathbb{P}^x[T_i^1<T_j^1]\cdot \mathbb{E}^x[\sigma_{p_{x,x_1,x_2,\ldots,x_n,i}}|T_i^1<T_j^1] \\
        & =-\deg(j)\cdot \sum_{x}\mathbb{P}^j[T_i^1<T_j^1]\cdot \mathbb{E}^j[\sigma_{p_{j,x,x_1,x_2,\ldots,x_n,i}}|T_i^1<T_j^1]                          \\
        & =-\deg(j)\cdot \mathbb{P}^j[T_i^1<T_j^1]\cdot \Omega_{ji}^1(j).
    \end{align*}
\end{proof}

\begin{proof}[Proof of Equation \eqref{eq:parallel}]
    We only show that $\mathcal{C}_{ji}^\sigma=\sum_{l=1}^m\mathcal{C}_{ji}^{\sigma,l}$.
    Note that for any given $l=1,\ldots,m$, one has that
    \begin{align*}
        \mathcal{C}_{ji}^{\sigma,l} & = -w_{jl}\mathbb{P}[T_i^1<T_j^1|X_0=j,X_1 = l]\,\mathbb{E}\left[\prod_{\ell=1}^{T_{j}^1}\sigma_{X_{\ell-1}X_{\ell}}|X_0=i,X_1 = l, T_i^1<T_j^1\right] \\
        & = -w_{jl}\mathbb{E}\left[\chi_i(T_{ij}^1)\prod_{\ell=1}^{T_{j}^1}\sigma_{X_{\ell-1}X_{\ell}}|X_0=i,X_1 = l\right].
    \end{align*}
    Hence,
    \begin{align*}
        \sum_{l=1}^m\mathcal{C}_{ji}^{\sigma,l} & = -\sum_{l=1}^mw_{jl}\mathbb{E}\left[\chi_i(T_{ij}^1)\prod_{\ell=1}^{T_{j}^1}\sigma_{X_{\ell-1}X_{\ell}}|X_0=i,X_1 = l\right]                           \\
        & = -\deg(j)\sum_{l=1}^m\mathbb{P}[X_1 = l|X_0=j]\mathbb{E}\left[\chi_i(T_{ij}^1)\prod_{\ell=1}^{T_{j}^1}\sigma_{X_{\ell-1}X_{\ell}}|X_0=i,X_1 = l\right] \\
        & = -\deg(j)\mathbb{E}\left[\chi_i(T_{ij}^1)\prod_{\ell=1}^{T_{j}^1}\sigma_{X_{\ell-1}X_{\ell}}|X_0=i\right]                                              \\
        & =-\deg(j)\mathbb{P}^j[T_i^1<T_j^1]\Omega_{ji}^1(j) =\mathcal{C}_{ji}^\sigma.
    \end{align*}
\end{proof}
%%%%%%%%%%%%%%%%%%%%%%%%%%%%%%%%%%%%%%%%%%%%%%%%%%%%%%%%%%%%%%%

\begin{proof}[Proof of \Cref{prop:resistance-equivalence}]
    By \Cref{prop:conductance matrix under equivalence} and \Cref{prop:mean path equivalence}, one has that
        \[F_{ij}\mathcal{C}^\sigma(i,j)=\mathcal{C}^\tau(i,j)F_{ij}\quad\text{ and}\]
        \[f(i)\Omega_{ij}^{\sigma,0}=\Omega_{ij}^{\tau,0}f(j),\quad f(j)\Omega_{ji}^{\sigma,0}=\Omega_{ji}^{\tau,0}f(i).\]
    The claim then follows from \Cref{prop:R = C} after direct calculation. 
\end{proof}

\begin{proof}[Proof of \Cref{prop:resistance-direct-sum}]
    This follows from \Cref{rmk:block-decomposition-direct-sum} and the observations made in the proof of \Cref{prop:resistance-equivalence}.
\end{proof}

\begin{proof}[Proof of \Cref{prop:resistance-lower-bound}]
    By \Cref{thm:resistance decomposition}, we only need to prove that $r_{ij}^\sigma\leq r_{ij}$ when $\sigma$ is absolutely inconsistent. In this case, we know by \Cref{thm:escape probability} that
    \[\mathcal{C}_{ii}^\sigma =c_{ij}I_{d\times d} +(\deg(i)-c_{ij})\cdot(I_{d\times d} - \Omega^1_{ii}(j))\] 
    By Jensen's inequality, we have that 
    \[\|\Omega_{ii}^1(j)\|_2 = \left\|\mathbb{E}^i\left[\prod_{\ell=1}^{T_{i}^1}\sigma_{X_{\ell-1}X_{\ell}}\middle|T_i^1<T_j^1\right]\right\|_2\leq  \mathbb{E}^i\left[\left\|\prod_{\ell=1}^{T_{i}^1}\sigma_{X_{\ell-1}X_{\ell}}\right\|_2\middle|T_i^1<T_j^1\right]= 1.\]
    This implies that $I_{d\times d}-\Omega_{ii}^1(j)$ is positive semidefinite. Since $\deg(i)\geq c_{ij}$, we have that all eigenvalues of $\mathcal{C}_{ii}^\sigma$ are lower bounded by $c_{ij}$. Therefore, we have that the eigenvalues of $(\mathcal{C}_{ii}^\sigma)^{-1}$ are upper bounded by $1/c_{ij}=r_{ij}$. In this way, by \Cref{thm:resistance decomposition} we have that 
    \[r_{ij}^\sigma =\frac{1}{2d}\left(\tr((\mathcal{C}_{ii}^\sigma)^{-1})+\tr((\mathcal{C}_{jj}^\sigma)^{-1})\right)\leq \frac{1}{2d}\cdot 2dr_{ij}=r_{ij}.\]
\end{proof}

\end{document}